\documentclass[smallcondensed]{svjour3}
\journalname{} 
\smartqed  
\usepackage{graphicx}
\usepackage{amssymb}
\usepackage{fancybox}

\usepackage{natbib}
\usepackage{times}
\usepackage{xcolor}
\usepackage{soul}
\usepackage[utf8]{inputenc}
\usepackage[small]{caption}
\usepackage{graphicx}
\usepackage{empheq}
\usepackage{wrapfig}
\usepackage[ruled, vlined, longend, linesnumbered]{algorithm2e}
\SetKwFor{For}{for }{ }{}
\usepackage{subfig}
\usepackage{graphicx}
\usepackage{array}
\usepackage{bm}

\newcommand{\nbAgents}{N}

\newtheorem{thr}{Theorem}

\newtheorem{defi}{Definition}

\newtheorem{lem}{Lemma}
\newtheorem{ex}{Example}

\newcommand{\GGI}{\mathtt{GGI}}
\newcommand{\rk}{\mathtt{rk}}

\newcommand{\rhoget}{\rho_{\mathtt{get}}}
\newcommand{\rhobreak}{\rho_{\mathtt{break}}}
\newcommand{\dmaxw}{d_{\max}^w}
\newcommand{\dmaxm}{d_{\max}^m}
\newcommand{\IN}{\mathtt{IN}_\mathbf{v}}
\newcommand{\OUT}{\mathtt{OUT}_\mathbf{v}}

\newcolumntype{C}[1]{>{\centering\let\newline\\\arraybackslash\hspace{0pt}}m{#1}}

\usepackage{tikz}
\usetikzlibrary{arrows,automata}
\usetikzlibrary{calc, chains, positioning}

\titlerunning{The GGI Stable Marriage Problem}
\title{Optimizing a Generalized Gini Index in Stable Marriage Problems: NP-Hardness, Approximation and a Polynomial Time Special Case}
\author{Hugo Gilbert \and Olivier Spanjaard}

\institute{H. Gilbert \and O. Spanjaard \at
              Sorbonne Université, CNRS, Laboratoire d'Informatique de Paris 6, LIP6, F-75005 Paris, France \\
              \email{\{hugo.gilbert,olivier.spanjaard\}@lip6.fr}           
}

\begin{document}
\maketitle

\begin{abstract} 
This paper deals with fairness in stable marriage problems. The idea studied here is to achieve fairness thanks to a Generalized Gini Index (GGI), a well-known criterion in inequality measurement, that includes both the egalitarian and utilitarian criteria as special cases. We show that determining a stable marriage optimizing a GGI criterion of agents' disutilities is an NP-hard problem. We then provide a polynomial time 2-approximation algorithm in the general case, as well as an exact algorithm which is polynomial time in the case of a constant number of non-zero weights parametrizing the GGI criterion.
\end{abstract}

\section{Introduction}

Since the seminal work of Gale and Shapley~[\citeyear{gale1962college}] on stable marriages, matching problems under preferences have been extensively studied both by economists and computer scientists. These problems involve two sets of agents (also called individuals in the sequel) that should be matched with each other while taking agents' preferences into account. The results obtained in the field have a tremendous number of applications, among which the National Resident Matching Program in the US (for allocating junior doctors to hospitals), the teacher allocation in France (for allocating newly tenured teachers to schools) or the allocation of lawyers in Germany (for assigning graduating lawyers to legal internship positions). For an overview of the applications of matching models under preferences, the interested reader can refer to a recent book chapter on this topic \citep{biro2017applications}.

The \emph{stable marriage problem} involves $n$ men and $n$ women, each of whom ranks the members of the opposite sex in order of preference. The goal is to find a \emph{stable} matching, i.e., a matching between men and women such that there is no man and woman that prefer each other to their current match. Gale and Shapley~[\citeyear{gale1962college}] provided an algorithm that computes a stable marriage. However, it is well-known that this algorithm favours one group (men or women, according to the way the algorithm is applied) over the other.

We are interested here in \emph{fair} stable marriage algorithms, i.e., in procedures favouring stable marriages that fairly share dissatisfactions --also called disutilities-- among individuals (irrespective of their sex), the dissatisfaction being defined for each woman (resp. man) as a function of the rank, in order of preferences, of the man (resp. woman) to whom she is paired with. Given the vector of individuals' dissatisfactions induced by a matching, there are several ways of formalizing the notion of ``fairness''. We mean here by fair stable marriage that the vector of individuals' dissatisfactions should be well-balanced. For example, consider the following instance of the stable marriage problem.

\begin{ex}
The instance consists of 10 men $\{m_1,\ldots,m_{10}\}$ and women $\{w_1,\ldots,w_{10}\}$ with the following preferences, where $i \succ^m_k j$ (resp. $i \succ^w_k j$) means that $m_k$ (resp. $w_k$) prefers $w_i$ to $w_j$ (resp. $m_i$ to $m_j$):

\vspace{-0.5cm}

{\small
\begin{center}
\begin{align*}
& m_1 : 1 \succ^m_1 2 \succ^m_1 3 \succ^m_1 4 \succ^m_1 5 \succ^m_1 6 \succ^m_1 7 \succ^m_1 8 \succ^m_1 9 \succ^m_1 10 \\ 
& m_2 : 2 \succ^m_2 1 \succ^m_2 3 \succ^m_2 4 \succ^m_2 5 \succ^m_2 6 \succ^m_2 7 \succ^m_2 8 \succ^m_2 9 \succ^m_2 10 \\
& m_3 : 3 \succ^m_3 1 \succ^m_3 2 \succ^m_3 4 \succ^m_3 5 \succ^m_3 6 \succ^m_3 7 \succ^m_3 8 \succ^m_3 9 \succ^m_3 10 \\
& m_4 : 7 \succ^m_4 1 \succ^m_4 2 \succ^m_4 3 \succ^m_4 6 \succ^m_4 4 \succ^m_4 5 \succ^m_4 8 \succ^m_4 9 \succ^m_4 10 \\
& m_5 : 6 \succ^m_5 1 \succ^m_5 2 \succ^m_5 3 \succ^m_5 7 \succ^m_5 4 \succ^m_5 5 \succ^m_5 8 \succ^m_5 9 \succ^m_5 10 \\
& m_6 : 4 \succ^m_6 1 \succ^m_6 2 \succ^m_6 3 \succ^m_6 5 \succ^m_6 7 \succ^m_6 6 \succ^m_6 8 \succ^m_6 9 \succ^m_6 10 \\
& m_7 : 5 \succ^m_7 1 \succ^m_7 2 \succ^m_7 3 \succ^m_7 4 \succ^m_7 7 \succ^m_7 6 \succ^m_7 8 \succ^m_7 9 \succ^m_7 10 \\
& m_8 : 8 \succ^m_8 4 \succ^m_8 5 \succ^m_8 6 \succ^m_8 10 \succ^m_8 7 \succ^m_8 1 \succ^m_8 2\succ^m_8 3 \succ^m_8 9 \\
& m_9 : 10 \succ^m_9 4 \succ^m_9 6 \succ^m_9 7 \succ^m_9 9 \succ^m_9 5 \succ^m_9 1 \succ^m_9 2\succ^m_9 3 \succ^m_9 8\\
& m_{10} : 9 \succ^m_{10} 4 \succ^m_{10} 5  \succ^m_{10} 7 \succ^m_{10} 8 \succ^m_{10} 6 \succ^m_{10} 1 \succ^m_{10} 2\succ^m_{10} 3 \succ^m_{10} 10
\end{align*}
\end{center}

\vspace{-0.8cm}

\begin{center}
\begin{align*}
&w_1:1 \succ^w_1 2 \succ^w_1 3 \succ^w_1 4 \succ^w_1 5 \succ^w_1 6 \succ^w_1 7 \succ^w_1 8 \succ^w_1 9 \succ^w_1 10 \\
&w_2:1 \succ^w_2 2 \succ^w_2 3 \succ^w_2 4 \succ^w_2 5 \succ^w_2 6 \succ^w_2 7 \succ^w_2 8 \succ^w_2 9 \succ^w_2 10 \\
&w_3:1 \succ^w_3 2 \succ^w_3 3 \succ^w_3 4 \succ^w_3 5 \succ^w_3 6 \succ^w_3 7 \succ^w_3 8 \succ^w_3 9 \succ^w_3 10 \\
&w_4:1 \succ^w_4 2 \succ^w_4 3 \succ^w_4 7 \succ^w_4 8 \succ^w_4 9 \succ^w_4 6 \succ^w_4 4 \succ^w_4 5 \succ^w_4 10 \\ 
&w_5:1 \succ^w_5 2 \succ^w_5 3 \succ^w_5 6 \succ^w_5 8 \succ^w_5 9 \succ^w_5 7 \succ^w_5 4 \succ^w_5 5 \succ^w_5 10 \\ 
&w_6:1 \succ^w_6 2 \succ^w_6 3 \succ^w_6 4 \succ^w_6 8 \succ^w_6 9 \succ^w_6 5 \succ^w_6 6 \succ^w_6 7 \succ^w_6 10 \\
&w_7:1\succ^w_7 2 \succ^w_7 3 \succ^w_7 5 \succ^w_7 8 \succ^w_7  9 \succ^w_7 4 \succ^w_7 6 \succ^w_7 7 \succ^w_7 10 \\ 
&w_8: 2 \succ^w_8 10 \succ^w_8 8 \succ^w_8 7 \succ^w_8 1 \succ^w_8 3 \succ^w_8 4 \succ^w_8 5 \succ^w_8 6 \succ^w_8 9 \\
&w_9: 1 \succ^w_9 2 \succ^w_9 9 \succ^w_9 10 \succ^w_9 3 \succ^w_9 4 \succ^w_9 5 \succ^w_9 6 \succ^w_9 7 \succ^w_9 8\\ 
&w_{10}:1 \succ^w_{10} 2 \succ^w_{10} 3 \succ^w_{10} 4 \succ^w_{10} 5 \succ^w_{10} 6 \succ^w_{10} 7 \succ^w_{10} 10 \succ^w_{10} 8 \succ^w_{10} 9 
\end{align*}
\end{center}
}

The stable marriages in this instance are:
\vspace{-0.1cm}
{\small
\begin{align*}
\mathbf{x}^1:~&  \{(m_1,w_1),(m_2,w_2),(m_3,w_3),(m_4,w_7),(m_5,w_6),(m_6,w_4),(m_7,w_5),\\&(m_8,w_8),(m_9,w_{10}),(m_{10},w_9)\}\\
\mathbf{x}^2:~&  \{(m_1,w_1),(m_2,w_2),(m_3,w_3),(m_4,w_7),(m_5,w_6),(m_6,w_5),(m_7,w_4),\\&(m_8,w_8),(m_9,w_{10}),(m_{10},w_9)\}\\
\mathbf{x}^3:~& \{(m_1,w_1),(m_2,w_2),(m_3,w_3),(m_4,w_6),(m_5,w_7),(m_6,w_4),(m_7,w_5),\\&(m_8,w_8),(m_9,w_{10}),(m_{10},w_9)\}\\
\mathbf{x}^4:~&  \{(m_1,w_1),(m_2,w_2),(m_3,w_3),(m_4,w_6),(m_5,w_7),(m_6,w_5),(m_7,w_4),\\&(m_8,w_8),(m_9,w_{10}),(m_{10},w_9)\}\\
\mathbf{x}^5:~&  \{(m_1,w_1),(m_2,w_2),(m_3,w_3),(m_4,w_6),(m_5,w_7),(m_6,w_5),(m_7,w_4),\\&(m_8,w_{10}),(m_9,w_9),(m_{10},w_8)\}\\
\end{align*}
}

\vspace{-0.9cm}

where a pair $(m_i,w_j)$ means that $m_i$ and $w_j$ are matched. 

\noindent If one assumes that the dissatisfaction of an individual is equal to the rank of the partner in his/her preference list, then the dissatisfactions induced by the previous stable marriages are:

{\small
\begin{center}
\begin{tabular}{ccC{1.75cm}C{1.75cm}}
\hline 
matching & vector of dissatisfactions & sum of$~~~~~~~~~~$ dissatisfactions & max of$~~~~~~~~~~$ dissatisfactions \\
\hline  
$\mathbf{x}^1$ & $(1,1,1,1,1,1,1,1,1,1,1,2,3,7,7,7,7,3,4,10)$ & 61 & 10\\
$\mathbf{x}^2$ & $(1,1,1,1,1,5,5,1,1,1,1,2,3,4,4,7,7,3,4,10)$ & 63 & 10\\
$\mathbf{x}^3$ & $(1,1,1,5,5,1,1,1,1,1,1,2,3,7,7,4,4,3,4,10)$ & 63 & 10\\
$\mathbf{x}^4$ & $(1,1,1,5,5,5,5,1,1,1,1,2,3,4,4,4,4,3,4,10)$ & 65 & 10\\
$\mathbf{x}^5$ & $(1,1,1,5,5,5,5,5,5,5,1,2,3,4,4,4,4,2,3,9)$ & 74 & 9\\
\hline
\end{tabular}
\end{center}
}

where the $i^{th}$ component of the vector is the dissatisfaction of $m_i$ for $i \in \{1,\ldots,10\}$, and of $w_{i-10}$ for $i \in \{11,\ldots,20\}$.

\label{exGILBERT}
\end{ex}

In this instance, the matching $\mathbf{x}^4$ can be considered as inducing a \emph{well-balanced} vector of dissatisfactions. The matchings $\mathbf{x}^1$, $\mathbf{x}^2$ and $\mathbf{x}^3$ indeed favour more some individuals (the men 
 in this case) than others, while matching $\mathbf{x}^5$ yields quite high dissatisfactions for numerous agents. The matching $\mathbf{x}^4$ is therefore a good compromise between the \emph{utilitarian} and the \emph{egalitarian} viewpoints, where the utilitarian viewpoint aims at minimizing the sum of dissatisfactions while the egalitarian viewpoint aims at minimizing the dissatisfaction of the worst off individual. Both the utilitarian and egalitarian approaches have been advocated for promoting fairness in the stable marriage problem \citep{gusfield1987three,gusfield1989stable}. Other approaches aim at treating equally men and women, by minimizing the absolute difference between the total dissatisfactions of the two groups (\emph{sex-equal} stable marriage problem \citep{kato1993complexity,mcdermid2014sex}) or by minimizing the maximum total dissatisfaction between the two groups (\emph{balanced} stable marriage problem \citep{manlove2013algorithmics}). However, note that, in the instance of Example~\ref{exGILBERT}, all these criteria favour either $\mathbf{x}^1$ (utilitarian) or $\mathbf{x}^5$ (egalitarian, sex-equal, balanced). Finally, there exists another type of approach, that is not based on assigning scores to marriages. In a first step, for each man, one lists all his possible matches in a stable marriage, in order of his preferences (this list includes as many elements as there are feasible stable marriages). In a second step, each man is matched with the median woman in the list. This procedure yields a stable marriage, which is called \emph{median} stable marriage \citep{teo1998geometry,cheng2010understanding}. In the instance of the example, the median stable marriage is $\mathbf{x}^4$. Nevertheless, in this article, we focus on determining a fair stable marriage by using a \emph{scoring rule}.
 
In social choice theory, a scoring rule assigns a score to each alternative by summing the scores given by every individual  over the alternative. This summation principle ensures that all individuals contribute equally to the score of an alternative. An alternative is usually a candidate in an election, but it can also be an element of a combinatorial domain. For instance, in proportional representation problems \citep{procaccia2008complexity}, where one aims at electing a committee, every feasible committee is an alternative. In the setting of stable marriage problems, every stable marriage is an alternative and the utilitarian approach is clearly a scoring rule where each individual evaluates a stable marriage by the rank of his/her match. An interesting extension of the class of scoring rules is the class of \emph{rank dependent scoring rules} \citep{goldsmith2014voting}, where, instead of limiting the aggregation to a summation operation, the scores are aggregated by taking into account their ranks in the ordered list of scores. As emphasized by \cite{goldsmith2014voting}, rank dependent scoring rules can be used to favour fairness by imposing some conditions on their parameters. A well known class of rank-dependent scoring rules in inequality measurement are the \emph{Generalized Gini Indices} (GGI) \citep{weymark1981generalized}. Furthermore, this class of rank dependent scoring rules circumvents both the utilitarian and egalitarian criteria. Their optimization on combinatorial domains have been studied in several settings (often under the name of \emph{Ordered Weighted Averages}): assignment problems \citep{lesca2018fair}, proportional representation \citep{elkind2015owa}, resource allocation \citep{heinen2015fairness}. To the best of our knowledge, the problem of determining a GGI optimal stable marriage has not been studied yet. This is precisely the purpose of the present work.

The paper is organized as follows. In Section~\ref{sec:notations}, we introduce notations and we formally define the GGI stable marriage problem studied here. Then, in Section~\ref{sec:complexity}, we prove that it is NP-hard to determine an optimal stable marriage according to a GGI criterion applied to agents' disutilities. In Section~\ref{sec:approximation}, we provide a polynomial time 2-approximation algorithm. Finally, in Section~\ref{sec:param}, we establish a parametrized complexity result with respect to a GGI-specific parameter.

\section{The GGI Stable Marriage Problem} \label{sec:notations}

Let $\mathcal{M}=\{m_1,\ldots,m_n\}$ denote the set of men, and $\mathcal{W}=\{w_1,\ldots,w_n\}$ the set of women. As in Example \ref{exGILBERT}, for each $m_k$ (resp. $w_k$), a preference relation $\succ^m_k$ (resp. $\succ^w_k$) is defined on $\mathcal{W}$ (resp. $\mathcal{M}$), where $i \succ^m_k j$ (resp. $i \succ^w_k j$) means that $m_k$ (resp. $w_k$) prefers $w_i$ to $w_j$ (resp. $m_i$ to $m_j$). We denote by $\rk(m_i,w_j)$ the rank of woman $w_j$ in the preference order of man $m_i$, and similarly for $\rk(w_j,m_i)$.

A solution of a stable marriage problem is a matching represented by a binary matrix $\mathbf{x}$, where $x_{ij}=1$ means that $m_i$ is matched with $w_j$. A matching $\mathbf{x}$ induces a matching function $\mu_\mathbf{x}$ defined by $w_j = \mu_\mathbf{x}(m_i)$ and $m_i = \mu_\mathbf{x}(w_j)$ if $x_{ij} = 1$. In a perfect matching (called indifferently matching or marriage from now on), every man (resp. woman) is matched with a different woman (resp. man). More formally, a matching is defined by: 
\begin{eqnarray}
\textstyle\sum_{i=1}^n x_{ij} = 1 \quad & \forall j \in \{1,\ldots,n \} \label{c_matchingi}\\
\textstyle\sum_{j=1}^n x_{ij} = 1 \quad & \forall i \in \{1,\ldots,n \} \label{c_matchingj}
\end{eqnarray}

A matching is said to be \emph{stable} if there exists no man and woman who prefer each other to their current partner. More formally, a perfect matching is stable if the following constraints hold \citep{vate1989linear}:
\begin{equation}
x_{ij} + \sum_{j' \succ^m_i j} x_{ij'} + \sum_{i' \succ^w_j i} x_{i'j} \ge 1 \quad \forall (i,j) \in \{1,\ldots,n\}^2 \label{c_stable}
\end{equation}
The set of stable marriages, i.e. binary matrices $\mathbf{x}$ such that constraints \ref{c_matchingi}, \ref{c_matchingj} and \ref{c_stable} hold, is denoted by $\mathcal{X}$. In their seminal paper, \cite{gale1962college} states that there always exists at least one stable marriage, which can be computed in $O(n^2)$. 
\\
The Gale-Shapley algorithm is based on a sequence of proposals from men to women. Each man proposes to the women following his preference order, pausing when a women agrees to be matched with him but continuing if his proposal is rejected. When a woman receives a proposal, she rejects it if she already has a better proposal according to her preferences. Otherwise, she agrees to hold it for consideration and rejects any former proposal that she might had. Such a sequence of proposals always leads to a stable marriage called man-optimal  stable marriage and denoted by $\mathbf{x}^m$ (if the role of men and women is reversed, we obtain the woman-optimal stable marriage denoted by $\mathbf{x}^w$). In the man-optimal stable marriage, each man has the best partner, and each woman has the worst partner, that is possible in any stable marriage. Contrarily, in the woman-optimal stable marriage, each woman has the best partner, and each man has the worst partner, that is possible in any stable marriage.

Two important properties of the Gale-Shapley algorithm are that:\\[0.5ex]
-- if $m$ proposes to $w$, then there is no stable marriage in which $m$ has a better match than $w$.\\ 
-- if $m$ proposes to $w$, then there is no stable marriage in which $w$ has a worse match than $m$.\\[0.5ex]
These properties justify the notion of preference shortlists obtained through the Gale-Shapley algorithm by removing any man $m$ from a woman $w$'s preference list and vice-versa, when $w$ receives a proposal from a man she prefers to $m$. Note that the shortlists that are obtained at the end of the algorithm do not depend on the order in which the proposals are made.
\begin{ex}
For instance, with the preferences of Example~\ref{exGILBERT}, the Gale-Shapley algorithm leads to the following shortlists:
{\small
\begin{align*}
m_1 &: 1 \succ^m_1 2 \succ^m_1 3 \succ^m_1 4 \succ^m_1 5 \succ^m_1 6 \succ^m_1 7 \succ^m_1 9 \succ^m_1 10 \\ 
m_2 &: 2 \succ^m_2 3 \succ^m_2 4 \succ^m_2 5 \succ^m_2 6 \succ^m_2 7 \succ^m_2 8 \succ^m_2 9 \succ^m_2 10 \\
m_3 &: 3 \succ^m_3 4 \succ^m_3 5 \succ^m_3 6 \succ^m_3 7 \succ^m_3 10 \\
m_4 &: 7 \succ^m_4 6 \succ^m_4 10 \\
m_5 &: 6 \succ^m_5 7 \succ^m_5 10 \\
m_6 &: 4 \succ^m_6 5 \succ^m_6 10 \\
m_7 &: 5 \succ^m_7 4 \succ^m_7 10 \\
m_8 &: 8 \succ^m_8 4 \succ^m_8 5 \succ^m_8 6 \succ^m_8 10 \succ^m_8 7\\
m_9 &: 10 \succ^m_9 4 \succ^m_9 6 \succ^m_9 7 \succ^m_9 9  \succ^m_9 5\\
m_{10} &: 9 \succ^m_{10} 8 \succ^m_{10} 10\\
\\
w_1 &: 1  \\
w_2 &: 1 \succ^w_2 2 \\
w_3 &: 1 \succ^w_3 2 \succ^w_3 3 \\
w_4 &: 1 \succ^w_4 2 \succ^w_4 3 \succ^w_4 7 \succ^w_4 8 \succ^w_4 9 \succ^w_4 6 \\ 
w_5 &: 1 \succ^w_5 2 \succ^w_5 3 \succ^w_5 6 \succ^w_5 8 \succ^w_5 9 \succ^w_5 7 \\  
w_6 &: 1 \succ^w_6 2 \succ^w_6 3 \succ^w_6 4 \succ^w_6 8 \succ^w_6 9 \succ^w_6 5 \\ 
w_7 &: 1 \succ^w_7 2 \succ^w_7 3 \succ^w_7 5 \succ^w_7 8 \succ^w_7 9 \succ^w_7 4 \\ 
w_8 &: 2 \succ^w_8 10 \succ^w_8 8 \\
w_9 &: 1 \succ^w_9 2 \succ^w_9 9 \succ^w_9 10\\
w_{10} &: 1 \succ^w_{10} 2 \succ^w_{10} 3 \succ^w_{10} 4 \succ^w_{10} 5 \succ^w_{10} 6 \succ^w_{10} 7 \succ^w_{10} 10 \succ^w_{10} 8 \succ^w_{10} 9 
\end{align*}
}
\end{ex}

These shortlists makes it possible to identify some transformations that can be applied from the man-optimal stable marriage to obtain other stable marriages (more favourable to women). These transformations are called rotations \citep{irving1986complexity}. A rotation is a sequence $\rho = (m_{i_0},w_{i_0}), \ldots,(m_{i_{r-1}},w_{i_{r-1}})$ of man-woman pairs such that, for each $i_k$ ($0\leq k \leq r-1$), (1) $w_{i_k}$ is first in $m_{i_k}$'s shortlist and (2) $w_{i_{k+1}}$ ($k+1$ taken modulo $r$) is second in $m_{i_k}$'s shortlist. Such a rotation is said to be exposed in the shortlists.

\begin{ex}
Continuing Example~\ref{exGILBERT}, there are two rotations exposed in the shortlists, $\rho_1 = (4,7),(5,6)$ and $\rho_2 = (6,4),(7,5)$. 
\end{ex}

Given a rotation, if each $m_{i_k}$ exchanges his current partner $w_{i_k}$ for $w_{i_{k+1}}$, then the matching remains stable. Eliminating a rotation $\rho = (m_{i_0},w_{i_0}), \ldots,(m_{i_{r-1}},w_{i_{r-1}})$ amounts to removing all successors $m$ of $m_{i_{k-1}}$ in $w_{i_k}$'s shortlist together with the corresponding appearances of $w_{i_k}$ in the shortlists of men $m$. The obtained stable marriage can then be read from the modified shortlists by matching each man with the first woman in his shortlist. In this new stable marriage, each woman (resp. man) is better off (resp. worse off) than before eliminating the rotation.

Once an exposed rotation has been identified and eliminated, then one or more rotations may be exposed in the resulting (further reduced) shortlists. This process may be repeated, and once all rotations have been eliminated, we obtain the woman optimal stable marriage. A rotation $\pi$ is said to be a predecessor of a rotation $\rho$, denoted by $\pi < \rho$, if $\rho$ cannot be exposed in the men shortlists before $\pi$ is eliminated. This notion of predecessors makes it possible to define what is called the rotation poset $(P, \leq)$ where $P$ is the set of all rotations and $\leq$ is the precedence relation that we have just mentioned. A closed set in a poset $(P, \leq)$ is a subset $R$ of $P$ such that $\rho \in R, \pi < \rho \Rightarrow \pi \in R$.

The following theorem is crucial to understand the importance of the rotation poset.

\begin{thr} \citep{irving1986complexity}
The stable marriages of a given stable marriage instance are in one-to-one correspondence with the closed subsets of the rotation poset.
\end{thr} 

In this correspondence, each closed subset $R$ represents the stable marriage obtained by eliminating the rotations in $R$ starting from $\mathbf{x}^m$.

The rotation poset can be represented as a  directed acyclic graph, with the rotations as nodes and an arc from $\pi$ to $\rho$ iff $\pi$ is an immediate predecessor of $\rho$ (i.e., $\pi < \rho$ and there is no rotation $\sigma$ such that $\pi < \sigma < \rho$). Note that this graph has at most $n(n-1)/2$ nodes, i.e., there are at most $n(n-1)/2$ rotations \citep{irving1987efficient}. Indeed, there are at most $n^2 -n$ pairs that can be involved in rotations (the $n$ pairs of $\mathbf{x}^w$ cannot be involved in a rotation). Each pair belong to at most one rotation and there are at least two pairs in each rotation. We will take advantage of the rotation poset in multiple places in the paper. Importantly, note that the rotation poset (actually a subgraph whose transitive closure is the rotation poset) can be generated in $O(n^2)$ \citep{gusfield1989stable}.


\begin{ex}
For instance, with the preferences of Example~\ref{exGILBERT}, the rotations and their immediate predecessors are given in the following table.
\begin{table}[!h]
\begin{center}
\begin{tabular}{|l|l|l|}
  \hline
  Rotation & New pairs & Immediate predecessors\\
  \hline
  $\rho_1 = (4,7),(5,6)$ & $(4,6), (5,7)$ &  \\
  $\rho_2 = (6,4),(7,5)$ & $(6,5), (7,4)$ &  \\
  $\rho_3 = (8,8),(9,10),(10,9)$ & $(8,10), (9,9), (10,8)$ &  $\rho_1,\rho_2$\\
  \hline
\end{tabular}
\end{center}
\end{table}

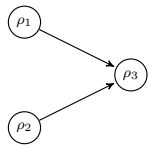
\begin{figure}[!h] 
   \centering
    \scalebox{.7}{\begin{tikzpicture}[->,>=stealth',shorten >=1pt,auto,node distance=3cm,
                    semithick]

  \node[circle,draw,text=black] (A)   at (0, 1)                 {$\rho_1$};
  \node[circle,draw,text=black] (B)   at (0,-1)                 {$\rho_2$};
  \node[circle,draw,text=black] (C)   at (2,0)                 {$\rho_3$};
  
  \path (A) edge  node {} (C)
	(B) edge  node {} (C);	  
	  \end{tikzpicture}}
    \caption{ \small Rotation poset in Example \ref{exGILBERT}.}
\end{figure}

\label{exIRVING3}
\end{ex}

This rotation poset shows that there are (potentially many) other stable marriages than the man-optimal or woman-optimal stable marriages. These other stable marriages are likely to be fairer than $\mathbf{x}^m$ and $\mathbf{x}^w$ as they are both extreme cases. 
In order to compute a fair stable marriage, the optimization of several aggregation functions has been investigated.\\ [0.5ex]
-- Utilitarian approach: $\sum_{i=1}^n \rk(m_i,\mu_\mathbf{x}(m_i))\!+\!\sum_{j=1}^n \rk(w_j,\mu_\mathbf{x}(w_j))$, which can be minimized in $O(n^3)$ \citep{Feder1994}).\\
-- Egalitarian approach: $\max \{\rk(p,\mu_\mathbf{x}(p)):p \in \mathcal{M}\cup \mathcal{W}\}$, which can also be minimized in $O(n^2)$ \citep{gusfield1987three}.\\
-- Sex-equal stable marriage:  $|\sum_{i=1}^n \rk(m_i,\mu_\mathbf{x}(m_i)) - \sum_{j=1}^n \rk(w_j,\mu_\mathbf{x}(w_j))|$, the minimization of which is NP-hard \citep{kato1993complexity}.\\ 
-- Balanced stable marriage: $\max \{\sum_{i=1}^n \rk(m_i,\mu_\mathbf{x}(m_i)) , \sum_{j=1}^n \rk(w_j,\mu_\mathbf{x}(w_j))\}$, the minimization of which is NP-hard \citep{manlove2013algorithmics}.\\[0.5ex]

Our contribution differs with previous works on the fair stable marriage problem. Indeed, we optimize a generalized Gini index on disutility values.\\ 


Given a matching $\mathbf{x}$, the \emph{disutility} $d(m_i,\mathbf{x})$ (also called \emph{dissatisfaction}) of a man $m_i$ is defined by $d(\rk(m_i,\mu_\mathbf{x}(m_i)))$, where $d:\mathbb{N}\rightarrow \mathbb{Q}^+$, is a strictly increasing function called disutility function. The disutility values $d(w_j,\mathbf{x})$ are defined similarly for women. 
Every stable marriage induces therefore a disutility vector:
\begin{equation*}
\mathbf{d}(\mathbf{x}) = (d(m_1,\mathbf{x}),\ldots,d(m_n,\mathbf{x}),d(w_1,\mathbf{x}),\ldots,d(w_n,\mathbf{x}))
\end{equation*}
with $\nbAgents = 2n$ components. Note that the use of disutility values (often called weights) is a common way to extend the traditional framework where the aggregation function is applied on rank values (see e.g., \cite{teo1998geometry,gusfield1989stable}). Using a unique disutility function for all agents guarantees that they all have the same importance in the aggregation operation. Indeed, the disutility values assigned to the ranks do not depend on the agent's identity. Note that both the egalitarian and the utilitarian variants of the stable marriage problem remain polynomially solvable if one uses disutility values.

\begin{ex}
We come back to Example~\ref{exGILBERT}. Let $d$ be the disutility function defined by $d(i) = (i-1)^2$, then the disutility values are given by the matrices $d_M$ and $d_W$ below where $d_M[i][j]$ (resp. $d_W[j][i]$) is the disutility of $m_i$ (resp. $w_j$) if he (resp. she) is matched with $w_j$ (resp. $m_i$). 

\scalebox{0.9}{\parbox{\columnwidth}{%
\begin{minipage}[c]{.46\linewidth}
\[
d_M: \begin{pmatrix}
0 & 1 & 4 & 9 & 16 & 25 & 36 & 49 & 64 & 81\\
1 & 0 & 4 & 9 & 16 & 25 & 36 & 49 & 64 & 81\\
1 & 4 & 0 & 9 & 16 & 25 & 36 & 49 & 64 & 81\\
1 & 4 & 9 & 25 & 36 & 16 & 0 & 49 & 64 & 81\\
1 & 4 & 9 & 25 & 36 & 0 & 16 & 49 & 64 & 81\\
1 & 4 & 9 & 0 & 16 & 36 & 25 & 49 & 64 & 81\\
1 & 4 & 9 & 16 & 0 & 36 & 25 & 49 & 64 & 81\\
36 & 49 & 64 & 1 & 4 & 9 & 25 & 0 & 81 & 16\\
36 & 49 & 64 & 1 & 25 & 4 & 9 & 81 & 16 & 0\\
36 & 49 & 64 & 1 & 4 & 25 & 9 & 16 & 0 & 81
\end{pmatrix}
\]
\end{minipage}\hfill
\begin{minipage}[c]{.46\linewidth}
\[
d_W: \begin{pmatrix}
0 & 1 & 4 & 9 & 16 & 25 & 36 & 49 & 64 & 81\\
0 & 1 & 4 & 9 & 16 & 25 & 36 & 49 & 64 & 81\\
0 & 1 & 4 & 9 & 16 & 25 & 36 & 49 & 64 & 81\\
0 & 1 & 4 & 49 & 64 & 36 & 9 & 16 & 25 & 81\\
0 & 1 & 4 & 49 & 64 & 9 & 36 & 16 & 25 & 81\\
0 & 1 & 4 & 9 & 36 & 49 & 64 & 16 & 25 & 81\\
0 & 1 & 4 & 36 & 9 & 49 & 64 & 16 & 25 & 81\\
16 & 0 & 25 & 36 & 49 & 64 & 9 & 4 & 81 & 1\\
0 & 1 & 16 & 25 & 36 & 49 & 64 & 81 & 4 & 9\\ 
0 & 1 & 4 & 9 & 16 & 25 & 36 & 64 & 81 & 49  
\end{pmatrix}
\]
\end{minipage}
}} \label{exIRVING3}

\end{ex}


Let $\mathbf{d} = (d_1,\ldots,d_N)$ denote a disutility vector. The generalized Gini index~\citep{weymark1981generalized} is defined as follows:
\begin{defi}
Let ${\bm \lambda} = (\lambda_1, \ldots, \lambda_\nbAgents)$ be a vector of weights such that $\lambda_1 \ge \ldots \ge \lambda_\nbAgents$. The $\GGI_{{\bm \lambda}}(\cdot)$ aggregation function induced by ${\bm \lambda}$ is defined by:
\begin{equation*}
\GGI_{{\bm \lambda}}(\mathbf{d}) = \sum_{i=1}^{\nbAgents} \lambda_{i} d^\downarrow_i,
\end{equation*} 
where $\mathbf{d^\downarrow}$ denotes the vector $\mathbf{d}$ ordered by nonincreasing values, i.e.,  $d^\downarrow_1 \geq d^\downarrow_2 \geq \ldots \geq d^\downarrow_\nbAgents$. 
\end{defi} 
The weights of the GGI aggregation function may be defined in a variety of manner. For instance, the weights initially proposed for the Gini social-evaluation function are:
\begin{equation} \label{gCoeff}
\lambda_i = (2(\nbAgents-i)+1)/\nbAgents^2 \quad \forall i \in \{1,\ldots,N\}
\end{equation}

\begin{ex}
Coming back to Example~\ref{exGILBERT}, if the weights ${\bm \lambda}$ are defined by Equation~\ref{gCoeff}  and the disutility function is defined by $d(i) = i$, the GGI values of the different stable marriages are (the lower the better):

{\small
\begin{center}
\begin{tabular}{ccC{1.75cm}}
\hline 
matching $\mathbf{x}$ & ordered vectors $\mathbf{d}^\downarrow(\mathbf{x})$ & $\GGI_{{\bm \lambda}}(\mathbf{d}(\mathbf{x}))$\\
\hline  
$\mathbf{x}^1$ & $(10, 7, 7, 7, 7, 4, 3, 3, 2, 1, 1, 1, 1, 1, 1, 1, 1, 1, 1, 1)$ & 4.4525 \\
$\mathbf{x}^2$ & $(10, 7, 7, 5, 5, 4, 4, 4, 3, 3, 2, 1, 1, 1, 1, 1, 1, 1, 1, 1)$ & 4.4725 \\
$\mathbf{x}^3$ & $(10, 7, 7, 5, 5, 4, 4, 4, 3, 3, 2, 1, 1, 1, 1, 1, 1, 1, 1, 1)$ & 4.4725 \\
$\mathbf{x}^4$ & $(10, 5, 5, 5, 5, 4, 4, 4, 4, 4, 3, 3, 2, 1, 1, 1, 1, 1, 1, 1)$ & 4.3925 \\
$\mathbf{x}^5$ & $(9 , 5, 5, 5, 5, 5, 5, 5, 4, 4, 4, 4, 3, 3, 2, 2, 1, 1, 1, 1)$ & 4.74 \\
\hline
\end{tabular}
\end{center}
}

\noindent We thus observe that using a GGI aggregation function makes it possible to obtain $\mathbf{x}^4$ as an optimal stable marriage.
\end{ex}

The GGI is also known in multicriteria decision making under the name of \emph{ordered weighted average}~\citep{yager1988ordered}. This aggregation function, to minimize, is well-known to satisfy the Pigou-Dalton transfer principle if $\lambda_1\!>\!\lambda_2\!>\!\ldots\!>\!\lambda_\nbAgents$:
\begin{defi}
An aggregation function $F$ satisfies the \emph{transfer principle} if for any $\mathbf{d} \in (\mathbb{R}^+)^\nbAgents$ and $\varepsilon \in (0,d_j-d_i)$ where $d_j>d_i$:\\
\begin{equation*}
F(d_1, \ldots, d_i + \varepsilon, \ldots, d_j - \varepsilon, \ldots , d_\nbAgents) < F(d_1, \ldots, d_\nbAgents).
\end{equation*}
\end{defi}
This condition states that the overall welfare should be improved by any transfer of disutility from a ``less happy'' agent $j$ to a happier agent $i$ given that this transfer reduces the gap between the disutilities of agent $i$ and $j$.
 We can now define the GGI Stable Marriage problem.\\[2ex] 
\fbox{\parbox{12cm}{
\textbf{GGI Stable Marriage (GGISM)}\\
\emph{INSTANCE:} Two disjoint sets of size $n$, the men and the women; for each person, a preference list containing all the members of the opposite sex; a vector of weight parameters ${\bm \lambda}$ and a disutility function $d$.\\
\emph{SOLUTION:} A stable marriage $\mathbf{x}$.\\
\emph{MEASURE:} $\GGI_{{\bm \lambda}}(\mathbf{d}(\mathbf{x}))$ (to minimize).
}}

\section{Complexity of the GGISM Problem} \label{sec:complexity}

The GGISM problem extends both the egalitarian and the utilitarian approaches to the stable marriage problem. Indeed, if the weights of the GGI operator are ${\bm \lambda} = (1,\ldots,1)$, one obtains the sum operation. If the weights are ${\bm \lambda} = (1,0,\ldots,0)$, one obtains the max operation. While both variants are polynomially solvable problems, the following result states that the GGISM problem is NP-hard:

\begin{thr}
The GGISM problem is NP-hard.
\end{thr}
\begin{proof}
We make a reduction from Minimum 2-Satisfiability, which is strongly NP-hard \citep{kohli1994minimum}.\\

\noindent \fbox{\parbox{12cm}{
\textbf{Minimum 2-Satisfiability (Min 2-SAT)}:\\
\emph{INSTANCE}: A set $V$ of variables, a collection $C$ of disjunctive clauses of at most 2 literals, where a literal is a variable or a negated variable in $V$.\\
\emph{SOLUTION}: A truth assignment for $V$.\\
\emph{MEASURE}: Number of clauses satisfied by the truth assignment (to minimize).
}}
\vspace{0.5cm}

To illustrate the reduction, we will use the following 2-SAT instance: 
\begin{align}
V &= \{v_1,v_2,v_3,v_4,v_5,v_6\} \label{form1a}\\
C &= \{ (v_1 \lor v_2), (\lnot v_2 \lor \lnot v_4), (\lnot v_1 \lor v_3), ( v_3 \lor \lnot v_4), v_2 , (v_5\lor v_6)\} \label{form1b}
\end{align}
As a preliminary step, note that we can get rid of variables that are present in only one clause. Such a variable is set to true if it is present as a negative literal in the clause and to false otherwise. It can then be removed from the instance. Furthermore, we can make sure that there are exactly two literals in each clause (by duplicating literals). For example, the instance described by Equations \ref{form1a} and \ref{form1b} can be modified to:
\begin{align}
V &= \{v_1,v_2,v_3,v_4\} \label{form2a}\\
C &= \{ (v_1 \lor v_2), (\lnot v_2 \lor \lnot v_4), (\lnot v_1 \lor v_3), ( v_3 \lor \lnot v_4), (v_2 \lor v_2)\} \label{form2b}
\end{align}
In the following we will denote by $n_v = |V|$ the number of variables and by $n_c = |C|$ the number of clauses. In the previous example $n_v = 4$ and $n_c = 5$. Furthermore, we will denote by $c_i$ the $i^{th}$ clause in $C$.\\

We are now going to create an instance of the GGISM problem such that: 
\begin{itemize}
\item There is a one-to-one correspondence between the stable marriages and the truth assignments for $V$. 
\item A stable marriage minimizing the GGI of the agent's disutilities corresponds to a truth assignment of $V$ minimizing the number of clauses that are satisfied.
\end{itemize}

In order to create a one-to-one correspondence between the stable marriages and the truth assignments for $V$, we are going to create a rotation $\rho_i$ for each variable $v_i \in V$. Each of these rotations will be exposed in the shortlists from the man-optimal stable marriage for the instance under construction. Additionally, we will ensure that these rotations will be the only ones of the stable marriage instance. In other words, the rotation poset will have one vertex per variable and no edge, as illustrated in Figure~\ref{fig:posetProof}.

\begin{figure}[!h] 
   \centering
    \scalebox{.7}{\begin{tikzpicture}[->,>=stealth',shorten >=1pt,auto,node distance=3cm, semithick]

  \node[circle,draw,text=black, text width = 0.5cm] (A)   at (0,2)                 {$~\rho_1$};
	\node[circle,draw,text=black,text width = 0.5cm] (B)   at (0,0)                 {$~\rho_2$};
  \node[circle,text=black] (C)   at (0,-1.5)                 {$\vdots$};
  \node[circle,draw,text=black,text width = 0.5cm] (E)   at (0, -3)                 {$\rho_{n_v}$};
\end{tikzpicture}}
\caption{\label{fig:posetProof}Rotation poset of the stable marriage instance generated by the reduction.}
\end{figure}
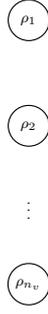

We now give the ``meaning'' of these rotations. Let's recall that in a stable marriage there is a one-to-one correspondence between the closed subsets of nodes of the rotation poset and the stable marriages. Now let $\mathbf{x}$ be a stable marriage corresponding to a closed subset $R$ of rotations, then the corresponding truth assignment over $V$ consists in setting $v_i = 1$ if $\rho_i\in R$ and $v_i = 0$ otherwise. Thus in the generated stable marriage instance, the man-optimal stable marriage (i.e., $R = \emptyset$) corresponds to a truth assignment where all variables in $V$ are set to 0 while the woman-optimal stable marriage (i.e., $R = \{\rho_i| v_i \in V\}$) corresponds to a truth assignment where all variables in $V$ are set to 1.

We now describe more precisely the fashion in which rotations $\rho_i$ are generated. For each variable $v_i$, we create a man-woman pair $(m_{ij}, w_{ij})$ for each clause $c_j$ that involves $v_i$ either as a positive or negative literal. If variable $v_i$ is present two times in a clause $c_j$, then two man-woman pairs $(m_{ij}, w_{ij})$ and $(m_{ij}', w_{ij}')$ are created. This induces the creation of $2n_c$ men and $2n_c$ women in the instance. The rotation $\rho_i$ then involves all the men and women induced by variable $v_i$. 
For example, in the instance described by Equations \ref{form2a} and \ref{form2b}, $\rho_2$ involves men $m_{21},m_{22},m_{25},m_{25}'$ and women $w_{21},w_{22},w_{25},w_{25}'$ as variable $v_2$ is present in $c_1, c_2$ and $c_5$. Let $r$ denote the number of times variable $v_i$ appears in $C$. The rotation $\rho_i$ is then induced by the following patterns in the shortlists of men $\{m_{ij}|v_i\in c_j\}$ and women  $\{w_{ij}|v_i\in c_j\}$:

\begin{minipage}[c]{.46\linewidth}
\begin{align*}
m_{ij_0} &: w_{ij_0} \succ^m_{ij_0} w_{ij_1} \\
m_{ij_1} &: w_{ij_1} \succ^m_{ij_1} w_{ij_2} \\
\vdots &\\
m_{ij_{r-2}} &: w_{ij_{r-2}} \succ^m_{ij_{r-2}} w_{ij_{r-1}} \\
m_{ij_{r-1}} &: w_{ij_{r-1}} \succ^m_{ij_{r-1}} w_{ij_0} 
\end{align*}
\end{minipage}
\begin{minipage}[c]{.46\linewidth}
\begin{align*}
w_{ij_0} &: m_{ij_{r-1}} \succ^w_{ij_0} m_{ij_0} \\
w_{ij_1} &: m_{ij_0} \succ^w_{ij_1} m_{ij_1} \\
\vdots &\\
w_{ij_{r-2}} &: m_{ij_{r-3}} \succ^w_{ij_{r-2}} m_{ij_{r-2}} \\
w_{ij_{r-1}} &: m_{ij_{r-2}} \succ^w_{ij_{r-1}} m_{ij_{r-1}} 
\end{align*}
\end{minipage}
\vspace{0.5cm}

For instance, rotation $\rho_2$ is induced by the following pattern in the shortlists:

\begin{minipage}[c]{.46\linewidth}
\begin{align*}
m_{21} &: w_{21} \succ^m_{21} w_{22} \\
m_{22} &: w_{22} \succ^m_{22} w_{25} \\
m_{25} &: w_{25} \succ^m_{25} w_{25}' \\
m_{25}'&: w_{25}' \succ^{m'}_{25} w_{21} 
\end{align*}
\end{minipage}
\begin{minipage}[c]{.46\linewidth}
\begin{align*}
w_{21}  &: m_{25}' \succ^w_{21} m_{21} \\
w_{22}  &: m_{21}  \succ^w_{22} m_{22} \\
w_{25}  &: m_{22}  \succ^w_{25} m_{25} \\
w_{25}' &: m_{25}  \succ^{w'}_{25} m_{25}' 
\end{align*}
\end{minipage} 
\vspace{0.5cm}

\begin{table}
\begin{center}
\begin{tabular}{cccc}
\hline
clause $c_j$ & in & out & decisive agents \\
\hline
$v_i \wedge v_k$ & & $\rho_i, \rho_k$ & $m_{ij}$, $m_{kj}$ \\
$v_i \wedge \neg v_k$ & $\rho_k$ &  $\rho_i$ & $m_{ij}$, $w_{kj}$ \\
$\neg v_i \wedge v_k$ & $\rho_i$ &  $\rho_k$ & $w_{ij}$, $m_{kj}$ \\
$\neg v_i \wedge \neg v_k$ & $\rho_i, \rho_k$ & & $w_{ij}$, $w_{kj}$ \\
$v_i \wedge v_i$ & & $\rho_i$ & $m_{ij}$, $m_{ij}'$ \\
$\neg v_i \wedge \neg v_i$ & $\rho_i$ & & $w_{ij}$, $w_{ij}'$ \\
\hline
\end{tabular}
\caption{\label{tab:Proof}Clause $c_j$ is not satisfied iff the rotations of the second (resp. third) column are included (resp. not included) in $R$. Consequently, clause $c_j$ is not satisfied iff the two agents in the last column are matched with their choices of rank $\rk^+(\cdot)$.}
\end{center}
\end{table}

Note that each man $m_{ij}$ or woman $w_{ij}$ is involved in one and only one rotation, which is $\rho_i$. As a consequence, each man or woman in the generated instance has only two possible matches in a stable marriage, namely $w_{ij_k}$ and $w_{ij_{k+1}}$ (modulo the size $r$ of rotation $\rho_i$) for $m_{ij_k}$, and $m_{ij_{k-1}}$ (modulo $r$) and $m_{ij_{k}}$ for $w_{ij_k}$. For simplicity, we will denote by $\rk^+(m_{ij})$ (resp. $\rk^-(m_{ij})$) the rank of the best (resp. worst) possible match for $m_{ij}$ in a stable marriage. Notations $\rk^+(w_{ij})$ and $\rk^-(w_{ij})$ are defined similarly for women.


Given a stable marriage characterized by a set $R$ of rotations, it is possible to determine if clause $c_j$ is satisfied by examining which rotations belong to $R$. According to the form of clause $c_j$, columns ``in'' and ``out'' of Table~\ref{tab:Proof} indicate which rotations should be included or not in $R$ so that $c_j$ is not satisfied. Assuming that $c_j$ involves variables $v_i$ and $v_k$ (or possibly their negations), it is sufficient to examine the matches of two specific agents among $m_{ij}, m_{kj}, w_{ij}, w_{kj}$ to determine if rotations $\rho_i$ and $\rho_k$ belong or not to $R$. These two specific agents are called \emph{decisive agents of $c_j$} in the following. We have indeed $\rho_i \in R$ iff the rank of the match of $m_{ij}$ is $\rk^+(m_{ij})$. Similarly, we have $\rho_i \not\in R$ iff the rank of the match of $w_{ij}$ is $\rk^+(w_{ij})$. Put another way, $m_{ij}$ (resp. $w_{ij}$) is a decisive agent of $c_j$ if $v_i$ (resp. $\neg v_i$) belongs to $c_j$. The clause $c_j$ is not satisfied iff the two decisive agents are with their match of rank $\rk^+(\cdot)$. The decisive agents according to the form of clause $c_j$ are given in the last column of Table~\ref{tab:Proof}.

For illustration, let us return to the 2-SAT instance described by Equations~\ref{form2a} and \ref{form2b}. Given the stable marriage instance generated by the reduction, and a stable marriage $\mathbf{x}$, clause $v_1\land v_2$ is \emph{not} satisfied iff $\rk(m_{11},\mu_\mathbf{x}(m_{11})) = \rk^+(m_{11})$ and $\rk(m_{21},\mu_\mathbf{x}(m_{21})) = \rk^+(m_{21})$. More generally, it is possible to count the number of clauses that are \emph{not} satisfied by examining the ranks of the matches of the decisive agents of each clause.



We will soon explain how to use a GGI operator to count the number of clauses that are not satisfied in the 2-SAT instance. Beforehand, we need to introduce fictitious agents in order to control the positions of the decisive agents in the ordered vector of disutilities for every stable marriage. More precisely, we introduce four fictitious agents $m_j$, $m_j'$, $w_j$, $w_j'$ per clause $c_j$ such that $m_j$ (resp. $m_j'$) is the first choice of $w_j$ (resp. $w_j'$) and vice-versa. Thus $m_j$ (resp. $m_j'$) can only be matched to $w_j$ (resp. $w_j'$) in a stable marriage, and therefore the fictitious agents will not interfere with the possible matches of the other agents. 

The fictitious agents are placed in the preference lists of the other agents such that $\rk^+(\cdot) = 2j+1$ and  $\rk^-(\cdot) = 2j+2$ for the two decisive agents of clause $c_j$. Furthermore, $\rk^+(\cdot) = 1$ and  $\rk^-(\cdot) = 2$ for the remaining (non-decisive) agents. Note that $2j+1>2$ as $j\ge 1$ and therefore the two decisive agents of $c_j$ are at positions $2(n_c-j)+1$ and $2(n_c-j)+2$ in the permutation that ranks the agents by non-increasing disutilities.

To achieve these properties, we position $2j$ fictitious agents at the beginning of the preference list of the decisive agents of clause $c_j$ (e.g., $m_1 m_1' \ldots m_j m_j'$ for a decisive agent $w_{ij}$). These agents are positioned just before the two possible matches of the agent in a stable marriage. Regarding the non-decisive agents, their two possible matches in a stable marriage are simply placed at the beginning of their preference lists.

For illustration, in the 2-SAT instance described by Equations~\ref{form2a} and \ref{form2b}, the preference list of agent $w_{22}$ (who is a decisive agent of $c_2$) is:
$$w_{22} : m_1 \succ_{22}^w m_1' \succ_{22}^w m_2 \succ_{22}^w m_2' \succ_{22}^w m_{21} \succ_{22}^w m_{22} \succ_{22}^w\ldots $$
and the preference list of agent $m_{22}$ (who is \emph{not} a decisive agent of $c_2$) is:
$$m_{22} : w_{22} \succ_{22}^m w_{25}\succ_{22}^m \ldots $$

This construction is illustrated in Figure \ref{prefexRed} (where symbols $\succ$ are omitted for readability reasons) for the Minimum 2-Satisfiability instance defined by Equations \ref{form2a} and \ref{form2b}. The preference lists of the agents are only partially given but note that they can be completed in any consistent way that would lead to complete and transitive orders.

\begin{figure}
\begin{minipage}[c]{.46\linewidth} 
\begin{align*}
m_{1}&: w_{1} \ldots\\ 
m_{1}'&: w_{1}' \ldots\\ 
\vdots\\ 
m_{5}&: w_{5} \ldots\\ 
m_{5}'&: w_{5}' \ldots\\[2ex]
m_{11}&: w_1 w_1' w_{11}  w_{13} \ldots\\
m_{13}&: w_1 w_1' w_2 w_2' w_3 w_3' w_{13} w_{11} \ldots\\[2ex]
m_{21}&: w_1 w_1' w_{21} w_{22} \ldots\\
m_{22}& : w_{22} w_{25} \ldots\\
m_{25}& : w_1 w_1' w_2 w_2' w_3 w_3' w_4 w_4' w_5 w_5' w_{25}  w_{25}' \ldots\\
m_{25}'&: w_1 w_1' w_2 w_2' w_3 w_3' w_4 w_4' w_5 w_5' w_{25}' w_{21} \ldots\\[2ex]
m_{33}&: w_1 w_1' w_2 w_2' w_3 w_3' w_{33} w_{34} \ldots\\
m_{34}&: w_1 w_1' w_2 w_2' w_3 w_3' w_4 w_4' w_{34} w_{33}\ldots\\[2ex]
m_{42}&: w_{42} w_{44}\ldots\\
m_{44}&: w_{44} w_{42}\ldots
\end{align*}
\end{minipage} \hfill
\begin{minipage}[c]{.46\linewidth}
\begin{align*}
w_{1}&: m_{1} \ldots\\ 
w_{1}'&: m_{1}' \ldots\\ 
\vdots\\
w_{5}&  : m_{5} \ldots\\ 
w_{5}'& : m_{5}' \ldots\\[2ex]  
w_{11}&:  m_{13}  m_{11} \ldots \\
w_{13}&:  m_{11}  m_{13} \ldots \\[2ex]
w_{21}&:  m_{25}' m_{21} \ldots\\
w_{22}& : m_1 m_1' m_2 m_2' m_{21} m_{22} \ldots\\
w_{25}&:  m_{22} m_{25} \ldots \\
w_{25}'&: m_{25} m_{25}' \ldots \\[2ex]
w_{33}&:  m_{34} m_{33} \ldots\\
w_{34}&:  m_{33} m_{34} \ldots\\[2ex]
w_{42}&:  m_1 m_1' m_2 m_2' m_{44} m_{42} \ldots\\
w_{43}&:  m_1 m_1' m_2 m_2' m_3 m_3' m_4 m_4' m_{42} m_{44} \ldots 
\end{align*}
\end{minipage}
\caption{\label{prefexRed}Preference lists obtained for the min 2-SAT instance of Equations~\ref{form2a} and \ref{form2b}.}
\end{figure}


We now explain how to define the disutility values attributed to each rank, as well as the weights of the GGI operator, so that the number of unsatisfied clauses can be inferred from the GGI value of the stable marriage.

\paragraph{Disutility values and weights of the GGI.} We first recall that each clause $c_j$ induces 6 agents that are matched either with their first or second choices and 2 agents (the decisive ones) that are matched with their choices of rank $2j+1$ or $2j+2$. By construction of the preference lists, note that no agent can be matched with a partner that is ranked strictly beyond $2n_c+2$ in his/her preference list. Therefore the values of $d(i)$ for $i > 2n_c+2$ play no role, and can be fixed arbitrarily as long as they are increasing with $i$ and strictly greater than $d(2n_c+2)$.\\
-- The increasing disutility values for ranks 1 to $2n_c+2$ are defined as follows (assuming that $n_c \ge 2$):
\begin{align*}
d(1) &= 0 \\
d(2) &= 1 \\
d(2j+1)  &= j + 1,\quad \forall j\in\{1,\ldots, n_c\}\\
d(2j+2)&= j + 1 + n_c^{-j}, \quad \forall j\in\{1,\ldots, n_c\}
\end{align*}
-- The non-increasing weights of the GGI are defined as follows: 
\begin{equation*}
{\bm \lambda} = (\underbrace{n_c^{n_c+1},n_c^{n_c},n_c^{n_c},n_c^{n_c-1}, \ldots, n_c^{3},n_c^{2},n_c^{2},n_c^{1}}_{2n_c\text{ weights}}, \underbrace{0, \ldots, 0}_{6n_c\text{ weights}}).
\end{equation*}
We recall that the $2n_c$ agents with the highest disutility values are the decisive agents (the two decisive agents of clause $c_j$ are matched with an agent of rank $2j+1$ or $2j+2$) and the $6n_c$ agents with the lowest disutility values are the \emph{non}-decisive agents (who are matched to one of their two first choices). Consequently, the weight vector ${\bm \lambda}$ attributes a weight 0 in the GGI operator to the $6n_c$ non-decisive agents while, for each clause $c_j$, it attributes a weight $n_c^{j}$ (resp. $n_c^{j+1}$) to the most satisfied (resp. least satisfied) of the two decisive agents of $c_j$.  


An upper bound on the GGI value is given by $\Delta_u = \sum_{j=1}^{n_c} (n_c^{j} + n_c^{j+1})d(2j+2)$. This would correspond to a stable marriage, where for each $c_{j}$, the two decisive agents of $c_j$ are both matched to their choice of rank $2 j+2$. Similarly, a lower bound on the GGI value is given by $\Delta_l =  \sum_{j=1}^{n_c} (n_c^{j} + n_c^{j+1})d(2j+1)$ (if the two  decisive agents of $c_j$ are both matched to their choice of rank $2 j+1$). Simple calculations show that $\Delta_l=\Delta_u-n_c(1+n_c)$. 

These bounds are useful for establishing Lemma \ref{lem:Proof} below, that makes it possible to infer the number of unsatisfied clauses from the GGI value. The lower the GGI value, the higher the number of unsatisfied clauses. Hence, minimizing the GGI value amounts to maximizing the number of unsatisfied clauses, which concludes the proof.

 

\begin{lem} \label{lem:Proof}
For the GGI stable marriage instance obtained by the method described above, a stable marriage $\mathbf{x}$ corresponds to a truth assignment on $V$ for which the number of unsatisfied clauses is:
$$
\left\lfloor \frac{\Delta_u - \GGI_{{\bm \lambda}}(\mathbf{d}(\mathbf{x}))}{n_c+1}\right\rfloor
$$ 
\end{lem}

\noindent\textit{Proof of Lemma \ref{lem:Proof}}. 
We wish to show that if a stable marriage $\mathbf{x}$ corresponds to a truth assignment on $V$ with exactly $k$ unsatisfied clauses then:
$$ \Delta_u - (k+1)(n_c+1) < \GGI_{{\bm \lambda}}(\mathbf{d}(\mathbf{x})) \le \Delta_u - k(n_c+1), $$
from which the lemma straightforwardly follows. 
 
Assume that $k$ clauses $\{c_{j_{1}},\ldots,c_{j_{k}}\}$ are unsatisfied for the truth assignment induced by $\mathbf{x}$. Then for each $c_{j_l}$, the two decisive agents of $c_{j_l}$ are both matched to their choice of rank $2 j_l+1$. Hence: 
\begin{align*}
\GGI_{{\bm \lambda}}(\mathbf{d}(\mathbf{x})) &\leq \Delta_u - \sum_{l=1}^k (n_c^{j_l}+n_c^{j_l+1})(d(2j_l+2) - d(2j_l+1))\\
&= \Delta_u - \sum_{l=1}^k (n_c^{j_l}+n_c^{j_l+1})n_c^{-j_l} = \Delta_u - k(n_c+1)
\end{align*}
because each decisive agent of clause $c_{j_l}$, for $l\in \{1,\ldots,k\}$, has a disutility of $d(2j_l+1)$ and not $d(2j_l+2)$.

Now, let $\{c_{j_{1}},\ldots,c_{j_{n_c-k}}\}$ denote the satisfied clauses for the truth assignment induced by $\mathbf{x}$. Then for each $c_{j_l}$, at least one of the two decisive agents of $c_{j_l}$ is matched to his/her choice of rank $2 j_l+2$. In the best case (w.r.t. the GGI value), only one of the two is matched to his/her choice of rank $2 j_l+2$, and his/her weight in the GGI aggregation is $n_c^{j_l+1}$ because his/her disutility is the highest among the two agents. Hence:
\begin{align*}
\GGI_{{\bm \lambda}}(\mathbf{d}(\mathbf{x}))&\geq \Delta_l+\sum_{l=1}^{n_c-k} n_c^{j_l+1} (d(2j_l+2)-d(2j_l+1))\\
& = \Delta_u-n_c(1+n_c) +\sum_{l=1}^{n_c-k} n_c^{j_l+1} n_c^{-j_l}\\
& = \Delta_u-n_c(1+n_c) + (n_c-k)n_c\\
& = \Delta_u-(k+1)n_c\\
& > \Delta_u-(k+1)(n_c+1)
\end{align*} 
This concludes the proof of the lemma.
\end{proof}

\section{A 2-approximation Algorithm} \label{sec:approximation}

We now present a polynomial time 2-approximation algorithm for the GGI stable marriage problem.

The 2-approximation algorithm uses a linear programming formulation of the stable marriage problem, based on the rotation poset \citep{gusfield1989stable}. It is indeed well-known that the set of stable marriages can be characterized by the following set of inequalities where we have one binary variable $y(\rho)$ for each rotation in the rotation poset and:
\begin{equation}
y(\rho') - y(\rho) \leq 0 \label{eqRotPos} 
\end{equation} 
for each pair of rotations such that $\rho$ precedes $\rho'$. Variable $y(\rho)$ is equal to $1$ if rotation $\rho$ is included in the closed set of rotations associated to the stable marriage and 0 otherwise. Importantly, note that the extreme points of the polytope defined by constraints \ref{eqRotPos} for $0 \le y(\rho) \le 1,~\forall \rho$ are in one-to-one correspondence with the stable marriages of the instance \citep{gusfield1989stable}. 
Furthermore, the stable marriage $\mathbf{x}$ characterized by variables $y(\rho)$ can be inferred by using Equations~\ref{eq_xy1},~\ref{eq_xy2} and~\ref{eq_xy3} below.

To explain this point we introduce some notations. Let $\Gamma$ denote the set of man-woman pairs included in at least one stable matching. These pairs can be found by looking at the pairs that are created and broken by each rotation. Indeed, note that for each pair $(m,w) \in \Gamma$, there exists exactly one rotation, denoted by $\rhoget(m,w)$, that creates this pair (unless this pair is in $\mathbf{x}^m$) and exactly one rotation, denoted by $\rhobreak(m,w)$, that breaks this pair (unless this pair is in $\mathbf{x}^w$).  
Then, one can compute variables $x_{ij}$ corresponding to a set of variables $y(\rho)$ by using the following equations:
\begin{align}
x_{ij} & = 1 - y(\rhobreak(i,j)) ,\quad \forall (i,j) \in \Gamma \mbox{ s.t. } x_{ij}^m =1  \label{eq_xy1} \\
x_{ij} & = y(\rhoget(i,j)) ,\quad \forall (i,j) \in \Gamma \mbox{ s.t. } x_{ij}^w =1 \label{eq_xy2} \\
x_{ij} & = y(\rhoget(i,j))-y(\rhobreak(i,j)) ,\quad \forall (i,j) \in \Gamma \mbox{ s.t. } x_{ij}^m = x_{ij}^w =0 \label{eq_xy3} 
\end{align} 

\begin{ex}
Let us come back to Example~\ref{exGILBERT}. The pairs in $\Gamma$ are listed in the left column of Table~\ref{tab:getbreak}. The rotations $\rhoget(m,w)$ and $\rhobreak(m,w)$ for each pair $(m,w)$ are given in the middle and right columns.

\begin{table}
\begin{center}
\begin{tabular}{|l|l|l|}
  \hline
  $(m,w) \in \Gamma$ & $\rhoget(m,w)$ & $\rhobreak(m,w)$\\
  \hline
  $(m_1,w_1)$ & & \\
  $(m_2,w_2)$ & & \\
  $(m_3,w_3)$ & & \\
  $(m_4,w_7)$ & & $\rho_1$  \\
  $(m_4,w_6)$ & $\rho_1$ & \\
  $(m_5,w_6)$ & & $\rho_1$\\
  $(m_5,w_7)$ & $\rho_1$ & \\
  $(m_6,w_4)$ & & $\rho_2$\\
  $(m_6,w_5)$ & $\rho_2$ & \\
  $(m_7,w_5)$ & & $\rho_2$\\
  $(m_7,w_4)$ & $\rho_2$ & \\
  $(m_8,w_8)$ & & $\rho_3$\\
  $(m_8,w_{10})$ & $\rho_3$ & \\
  $(m_9,w_{10})$ & & $\rho_3$\\
  $(m_9,w_9)$ & $\rho_3$ & \\
  $(m_{10},w_9)$ & & $\rho_3$\\
  $(m_{10},w_8)$ & $\rho_3$ & \\
  \hline
\end{tabular}
\caption{\label{tab:getbreak}Rotations $\rhoget(m,w)$ and $\rhobreak(m,w)$ in Example~\ref{exGILBERT}.}
\end{center}
\end{table}
\end{ex}



A mathematical programming formulation of the GGISM problem reads as follows:
\begin{empheq}[left=\mathcal{P}
\empheqlbrace]{align} 
& \min_{\mathbf{d},\mathbf{x},\mathbf{y}} \GGI_{{\bm \lambda}}(\mathbf{d}) \notag \\
d^m_i & = \sum_{(i,j) \in \Gamma} x_{ij} d(\rk(m_i,w_j)) ,\quad \forall i \in \{1,\ldots,n\} \notag \\
d^w_j & = \sum_{(i,j) \in \Gamma} x_{ij} d(\rk(w_j,m_i)) ,\quad \forall j \in \{1,\ldots,n\} \notag \\
x_{ij} & = 1 - y(\rhobreak(i,j)) ,\quad \forall (i,j) \in \Gamma \mbox{ s.t. } x_{ij}^m =1 \label{const:x-y1} \\
x_{ij} & = y(\rhoget(i,j)) ,\quad \forall (i,j) \in \Gamma \mbox{ s.t. } x_{ij}^w =1  \label{const:x-y2} \\
x_{ij} & = y(\rhoget(i,j))-y(\rhobreak(i,j)) ,\quad \forall (i,j) \in \Gamma \mbox{ s.t. } x_{ij}^m = x_{ij}^w =0 \label{const:x-y3} \\
y(\rho') & - y(\rho) \le 0 ,\quad \forall (\rho,\rho') \mbox{ s.t. } \rho < \rho' \label{const:y} \\
d^m_i & \ge 0 ,\quad \forall i \in \{1,\ldots,n\} \notag\\
d^w_j & \ge 0 ,\quad \forall j \in \{1,\ldots,n\} \notag\\
x_{ij} & \ge 0 ,\quad \forall (i,j) \in \Gamma\notag\\ 
y(\rho) & \in \{0,1\},\quad \forall \rho \in P \notag
\end{empheq}

\noindent where $d_i^m$ (resp. $d_j^w$) represents the disutility of $m_i$ (resp. $w_j$), $\mathbf{d} = (d_1^m,\ldots,d_n^m,d_1^w,\ldots,d_n^w)$, and as usual:
\begin{itemize}
\item  $x_{ij}=1$ (resp. 0) if $(m_i,w_j)$ is (resp. is not) in the stable marriage $\mathbf{x}$,
\item $y(\rho)=1$ (resp. 0) if $\rho$ belongs (resp. does not) to the set of rotations characterizing $\mathbf{x}$,
\item  $P$ is the set of all rotations. 
\end{itemize}
Let us denote by $\widehat{\mathcal{P}}$ the linear programming relaxation of $\mathcal{P}$ where $y(\rho) \in \{0,1\}$ is replaced by $0 \le y(\rho) \le 1$. Importantly, note that variables $y(\rho)$ in an optimal solution to $\widehat{\mathcal{P}}$ are not necessarily integer because the objective function is non-linear (and therefore there does not necessarily an optimal vertex in the solution polytope).


A polynomial time 2-approximation algorithm can be obtained by rounding an optimal solution of $\widehat{\mathcal{P}}$. The 2-approximation algorithm writes as follows:\\ [1ex]
{\sc Rounding Algorithm}
\begin{enumerate}
\item Solve $\widehat{\mathcal{P}}$ and let $(\mathbf{\hat{d}},\mathbf{\hat{x}},\mathbf{\hat{y}})$ denote an optimal solution to $\widehat{\mathcal{P}}$; 
\item For each $\rho \in P$, set $y(\rho) = 1$  if $\hat{y}(\rho) \ge 0.5$, and $y(\rho) = 0$ otherwise;
\item Return the stable marriage $\mathbf{x}$ obtained from $\mathbf{y}$ by using constraints~\ref{const:x-y1}--\ref{const:x-y3} in $\mathcal{P}$.
\end{enumerate}

\begin{ex}
Coming back to Example \ref{exGILBERT}, assume that the weights of the GGI operator are defined by Equation~\ref{gCoeff} and that the disutility function is defined by $d(i) = i$. Then, an optimal solution $(\mathbf{\hat{d}},\mathbf{\hat{x}},\mathbf{\hat{y}})$ to $\widehat{\mathcal{P}}$  is characterized by $\hat{y}(\rho_1) = \hat{y}(\rho_2) = 0.75$ and $\hat{y}(\rho_3) = 0$ (for a GGI value of $4.3075$). For this instance, by {\sc Rounding Algorithm}, the obtained vector $\mathbf{y}$ is therefore $y(\rho_1) = y(\rho_2) = 1$ and $y(\rho_3) = 0$. This corresponds to stable marriage $\mathbf{x}^4$, which is in fact an optimal solution. 
\end{ex}

Steps 2 and 3 of the algorithm can obviously be performed in polynomial time.
In step 1, solving $\widehat{\mathcal{P}}$ can also be performed in polynomial time by using one of the linearizations of the GGI operator proposed by \cite{ogryczak2003solving}.
 The following lemma ensures that the returned solution is a 2-approximation of an optimal solution of $\mathcal{P}$:
\begin{lem}
For any feasible solution $(\mathbf{\hat{d}},\mathbf{\hat{x}},\mathbf{\hat{y}})$ of $\widehat{\mathcal{P}}$, the feasible solution $(\mathbf{d},\mathbf{x},\mathbf{y})$ of $\mathcal{P}$ obtained by setting 
$$ y(\rho) = \left \lbrace \begin{array}{cl}
            1  & \mbox{ if } \hat{y}(\rho) \ge 0.5,\\ 
            0  & \mbox{ otherwise}
            \end{array} \right.
$$            
 is  such that 
$\mathbf{\hat{d}} \ge \frac{1}{2}\mathbf{d}$ where $\ge$ is taken componentwise.  \label{lemhalf}
\end{lem}
\begin{proof}
In order to establish the result stated in the lemma, we introduce the notion of man and woman \emph{weights} of a rotation. Given a rotation $\rho = (m_{i_0}, w_{i_0}),\ldots,(m_{i_{r-1}},w_{i_{r-1}})$ we define the $m_i$-weight of that rotation by:
\begin{equation*}
\omega_i^m(\rho) = \left\{\begin{array}{l}
d(\rk(m_{i_k},w_{i_k})) - d(\rk(m_{i_k},w_{i_{(k+1)\mbox{\tiny ~mod } r}})) \mbox{ if } i \in \{i_0,\ldots,i_{r-1}\} \mbox{ and } i=i_k\\
0 \mbox{ otherwise.} \end{array}\right.  
\end{equation*} 
Similarly, we define the $w_j$-weight of that rotation by:
\begin{equation*}
\omega_j^w(\rho) = \left\{\begin{array}{l}
d(\rk(w_{i_k},m_{i_k})) - d(\rk(w_{i_k},m_{i_{(k-1)\mbox{\tiny ~mod } r}})) \mbox{ if } j \in \{i_0,\ldots,i_{r-1}\} \mbox{ and } j=i_k\\
0 \mbox{ otherwise.} \end{array}\right.  
\end{equation*}  

Note that a man weight of a rotation will always be negative while a woman weight of a rotation will always be positive.  

Assume that $\rho$ is a rotation that is exposed in a stable marriage $\mathbf{x}$, and let $\mathbf{x}'$ be the stable marriage obtained from $\mathbf{x}$ by eliminating $\rho$. Then:
\begin{equation*}
d(m_i,\mathbf{x}') = d(m_i,\mathbf{x}) - \omega_i^m(\rho), \hspace{0.5cm} d(w_j,\mathbf{x}') = d(w_j,\mathbf{x}) - \omega_j^w(\rho).
\end{equation*}
Consequently, if $\mathbf{x}$ is the stable marriage obtained from the man-optimal stable marriage $\mathbf{x}^m$ by eliminating rotations $\rho_1, \ldots, \rho_t$, then:
\begin{equation*}
d(m_i,\mathbf{x}) = d(m_i,\mathbf{x}^m) - \sum_{k=1}^t \omega_i^m(\rho_k), \hspace{0.5cm} d(w_j,\mathbf{x}) = d(w_j,\mathbf{x}^m) - \sum_{k=1}^t \omega_j^w(\rho_k).
\end{equation*}
We now establish the result stated in the lemma. Let $(\mathbf{\hat{d}},\mathbf{\hat{x}},\mathbf{\hat{y}})$ denote a feasible solution of $\widehat{\mathcal{P}}$. The previous equations extend as follows for solutions of $\widehat{\mathcal{P}}$:
\begin{equation}
\hat{d}^m_i = d(m_i,\mathbf{x}^m) - \sum_{\rho \in P} \hat{y}(\rho)\omega_i^m(\rho), \hspace{0.5cm} \hat{d}^w_j = d(w_j,\mathbf{x}^m) - \sum_{\rho \in P} \hat{y}(\rho) \omega_j^w(\rho). \label{eq:hatd^m_i}
\end{equation}

Now consider the feasible solution $(\mathbf{d},\mathbf{x},\mathbf{y})$ of $\mathcal{P}$ defined by $y(\rho) = 1$ if $\hat{y}(\rho) \ge 0.5$, and $0$ otherwise. The feasibility of $(\mathbf{d},\mathbf{x},\mathbf{y})$ comes from the fact that $\{\rho : \hat{y}(\rho) \ge 0.5\}$ is a closed set of rotations. Indeed, note that constraints~\ref{const:y} ensures that $y(\rho') \le y(\rho)$ for all $\rho<\rho'$. We have:
\begin{align*}
\hat{d}^m_i - \frac{1}{2}d^m_i &= d(m_i,\mathbf{x}^m) - \sum_{\rho \in P}\hat{y}(\rho) \omega_i^m(\rho) -    \frac{1}{2} (d(m_i,\mathbf{x}^m) - \sum_{\rho \in P} y(\rho) \omega_i^m(\rho))\\
	    &=\frac{1}{2} d(m_i,\mathbf{x}^m) - \sum_{\rho\in P} (\hat{y}(\rho) - \frac{1}{2}y(\rho)) \omega_i^m(\rho)
	    \\&\ge \frac{1}{2} d(m_i,\mathbf{x}^m) \ge 0
\end{align*}
as $0 \le (\hat{y}(\rho) - \frac{1}{2}y(\rho))$ for all $\rho\in P$ and $\omega_i^m(\rho) \le 0$ for all $i\in\{1,\ldots,n\}$ and $\rho\in P$. Hence, $\hat{d}^m_i \ge \frac{1}{2}d^m_i$ for all $i\in\{1,\ldots,n\}$.

Similarly, for women we have:
\begin{align*}
\hat{d}^w_j-\frac{1}{2}d^w_j &=d(w_j,\mathbf{x}^m) - \sum_{\rho\in P} \hat{y}(\rho) \omega_j^w(\rho)-\frac{1}{2}(d(w_j,\mathbf{x}  ^m) - \sum_{\rho \in P} y(\rho) \omega_j^w(\rho))\\
	    &=\frac{1}{2} d(w_j,\mathbf{x}^m) - \sum_{\rho\in P} (\hat{y}(\rho) - \frac{1}{2}y(\rho)) \omega_j^w(\rho)\\
	    &\ge \frac{1}{2}(d(w_j,\mathbf{x}^m) - \sum_{\rho \in P} \omega_j^w(\rho))
\end{align*}
as $(\hat{y}(\rho) - \frac{1}{2}y(\rho)) \le 0.5$ for all $\rho\in P$ and $\omega_j^w(\rho) \ge 0$ for all $j\in\{1,\ldots,n\}$ and $\rho\in P$. 
Since eliminating all rotations from $\mathbf{x}^m$ leads to $\mathbf{x}^w$, we have that $\frac{1}{2}(d(w_j,\mathbf{x}^m) - \sum_{\rho \in P} \omega_j^w(\rho)) = \frac{1}{2} d(w_j,\mathbf{x}^w)$. Therefore, $\hat{d}^w_j-\frac{1}{2}d^w_j \geq 0$ and hence, $\hat{d}^w_j \ge \frac{1}{2}d^w_j$ for all $j\in\{1,\ldots,n\}$. 

By combining the inequalities obtained for men and women, we obtain that  $\mathbf{\hat{d}} \ge \frac{1}{2}\mathbf{d}$, which concludes the proof.
\end{proof}

We can now state the main result of this section:

\begin{thr}
{\sc Rounding Algorithm} is a polynomial time 2-approximation algorithm for the GGI stable marriage problem, and the bound is tight.
\end{thr}

\begin{proof}
We first recall that all steps of {\sc Rounding Algorithm} can be performed in polynomial time. Furthermore, by Lemma~\ref{lemhalf}, the feasible solution $(\mathbf{{d}},\mathbf{{x}},\mathbf{{y}})$ generated by {\sc Rounding Algorithm} is such that $\mathbf{\hat{d}} \ge \frac{1}{2}\mathbf{d}$, where $(\mathbf{\hat{d}},\mathbf{\hat{x}},\mathbf{\hat{y}})$ is an optimal solution to $\widehat{\mathcal{P}}$. Consequently:
\[
\GGI_{{\bm \lambda}}(\mathbf{\hat{d}}) \ge \frac{1}{2} \GGI_{{\bm \lambda}}(\mathbf{d})
\]
because $\GGI_{{\bm \lambda}}(\mathbf{{d}}) \le \GGI_{{\bm \lambda}}(\mathbf{{d'}})$ for $\mathbf{{d}} \le \mathbf{{d'}}$ (see e.g. \cite{fodor1995characterization}) and $\GGI_{{\bm \lambda}}(\alpha \mathbf{d})=\alpha\GGI_{{\bm \lambda}}(\mathbf{d})$ for $\alpha > 0$.

For the tightness of the bound, consider the following instance of the stable marriage problem:
\begin{center}
$m_1 : 1 \succ^m_1 2 \succ^m_1 3$ \hspace{0.5cm} $w_1 : 2 \succ^w_1 3 \succ^w_1 1$\\ 
$m_2 : 2 \succ^m_2 3 \succ^m_2 1$ \hspace{0.5cm} $w_2 : 3 \succ^w_2 1 \succ^w_2 2$\\
$m_3 : 3 \succ^m_3 1 \succ^m_3 2$ \hspace{0.5cm} $w_3 : 1 \succ^w_3 2 \succ^w_3 3$
\end{center}
There are two rotations $\rho_1 =(1,1),(2,2),(3,3)$ and $\rho_2=(1,2),(2,3),(3,1)$, with $\rho_1 < \rho_2$, which yield three stable marriages: 
\begin{itemize}
\item the man-optimal stable marriage $\mathbf{x}^m$ in which each man is matched with his first choice, and which corresponds to eliminating no rotation,
\item the woman-optimal stable marriage $\mathbf{x}^w$ in which each woman is matched with her first choice, and which corresponds to eliminating both rotations,
\item a ``compromise'' stable marriage $\mathbf{x}^c$ in which each agent is matched with his/her second choice, and which corresponds to eliminating only $\rho_1$.
\end{itemize}
We use the disutility function $d$ defined by $d(1)=0$, $d(2)=1+\epsilon$ and $d(3)=2$ (with $\epsilon>0$) and the following GGI weights: ${\bm \lambda} = (a,b,c,0,0,0)$ with $a \ge b \ge c > 0$. The disutility vectors of the three stable marriages are then $\mathbf{d}(\mathbf{x}^m) = (0,0,0,2,2,2)$, $\mathbf{d}(\mathbf{x}^w) = (2,2,2,0,0,0)$ and $\mathbf{d}(\mathbf{x}^c) = (1+\epsilon,1+\epsilon,1+\epsilon,1+\epsilon,1+\epsilon,1+\epsilon)$.


By using Equation~\ref{eq:hatd^m_i} from the proof of Lemma~\ref{lemhalf}, the value of $\hat{d}^m_i$ for each man $m_i$ and the value of $\hat{d}^w_j$ for each woman $w_j$ are written as follows in terms of $\hat{y}(\rho_1)$ and $\hat{y}(\rho_2)$:
\begin{align*}
\hat{d}^m_i &= 0 + \hat{y}(\rho_1) (1 + \epsilon) + \hat{y}(\rho_2) (1 - \epsilon) \quad \forall i\in\{1,2,3\}\\
\hat{d}^w_j &= 2 - \hat{y}(\rho_1) (1 - \epsilon) - \hat{y}(\rho_2) (1 + \epsilon) \quad \forall j\in\{1,2,3\}
\end{align*}
where $\hat{y}(\rho_1) \ge \hat{y}(\rho_2)$. We see that the three men share the same disutility value, as well as the three women. The GGI value of a feasible solution to $\widehat{\mathcal{P}}$ is thus completely determined by the common disutility of the men or the common disutility of the women because only the three least satisfied agents are taken into account in ${\bm \lambda}$. Consequently, an optimal solution to $\widehat{\mathcal{P}}$ minimizes $\max\{\hat{d}_1^m,\hat{d}_1^w\}$ for $\hat{y}(\rho_1) \ge \hat{y}(\rho_2)$. Simple calculations make it possible to conclude that the only optimal solution to $\widehat{\mathcal{P}}$ is characterized by $\hat{y}(\rho_1)=\hat{y}(\rho_2)=0.5$. 

For this instance, {\sc Rounding Algorithm} returns therefore the woman-optimal stable marriage which has the ordered disutility vector $(2,2,2,0,0,0)$ and a GGI value of $2(a + b + c)$. However, an optimal stable marriage is the ``compromise'' stable marriage, which has the ordered disutility vector $(1+\epsilon,1+\epsilon,1+\epsilon,1+\epsilon,1+\epsilon,1+\epsilon)$ and a GGI value of $(1 + \epsilon)(a + b + c)$. By taking the limit for $\epsilon$ going to 0, we obtain the tightness of the bound.  
\end{proof}

\begin{remark}
Note that the approach taken in {\sc Rounding Algorithm} is valid for any aggregation criterion $F$ on dissatisfactions of agents for which the following condition holds:
\begin{align*}
F(\frac{1}{2}\mathbf{d}(\mathbf{x})) &= \frac{1}{2}F(\mathbf{d}(\mathbf{x}))\\
F(\mathbf{d}(\mathbf{x}) + \mathbf{r}) &\ge F(\mathbf{d}(\mathbf{x}))
\end{align*}
where $\mathbf{r}$ is any non-negative vector.
\end{remark}

\begin{remark}
A general approximation result for the optimization of a generalized Gini index in muliobjective optimization problems has been proposed by \cite{kasperski2015combinatorial}. For the GGISM problem, it amounts to compute an optimal stable marriage according to the sum of disutilities of pairs $(m_i,w_j)$, where the disutility of a pair $(m_i,w_j)$ is defined by $\lambda_1 \max\{d(\rk(m_i,w_j)),d(\rk(w_j,m_i))\} + \lambda_2 \min\{d(\rk(m_i,w_j)),d(\rk(w_j,m_i))\}$. This can be performed in polynomial time by linear programming. The returned solution is a $N\lambda_1$-approximation, provided $\sum_{i=1}^N \lambda_i = 1$. To obtain a better guarantee than 2, one should have $\lambda_1 < 2/N$. On the contrary, by taking advantage of the specific structure of the stable marriage problem, our approach yields a 2-approximation \emph{whatever weights are used}.
\end{remark}

\section{The GGI Stable Marriage Problem with a Bounded Number of Non-zero Weights}
\label{sec:param}

In this section, we provide an algorithm whose complexity is $O(2^Kn^{2K+4})$ where $K = \max\{i : \lambda_i > 0\}$. Hence, the complexity is polynomial time if $K$ is assumed to be a constant, where $K$ is the number of non-zero weights in the GGI operator. In the parametrized complexity terminology \citep{niedermeier2006invitation}, this means that the GGI stable marriage problem belongs to class XP for parameter $K$.

We adopt a brute-force approach to solve the problem in $O(2^Kn^{2K+4})$. Let
\begin{multline*}
\mathbf{t}(\mathbf{x})=((d(m_1,\mathbf{x}),m_1,\mu_{\mathbf{x}}(m_1)),\ldots,(d(m_n,\mathbf{x}), m_n,\mu_{\mathbf{x}}(m_n)),\\ (d(w_1,\mathbf{x}),w_1,\mu_{\mathbf{x}}(w_1)),\ldots,(d(w_n,\mathbf{x}),w_n,\mu_{\mathbf{x}}(w_n))
\end{multline*}
denote the vectors of triples $(d(a_i,\mathbf{x}),a_i,\mu_{\mathbf{x}}(a_i))$ induced by stable marriage $\mathbf{x}$, where $d(a_i,\mathbf{x})$ is the dissatisfaction of agent $a_i$ when matched with $\mu_{\mathbf{x}}(a_i)$. We denote by $T^\downarrow(\mathbf{x})$ the set of vectors $\mathbf{t}^\downarrow(\mathbf{x})$ that can be obtained from  $\mathbf{t}(\mathbf{x})$ by sorting the triples in decreasing order of dissatisfactions. 
The projection of a vector $\mathbf{t}^\downarrow(\mathbf{x}) \in  T^\downarrow(\mathbf{x})$ on the $K$ first components is denoted by $\mathbf{t}^\downarrow_K(\mathbf{x})$. We denote by  $T^\downarrow_K(\mathbf{x})$ the set $\{\mathbf{t}^\downarrow_K(\mathbf{x}) : \mathbf{t}^\downarrow(\mathbf{x}) \in  T^\downarrow(\mathbf{x})\}$.

For instance, assume that $\mathbf{t}(\mathbf{x})$ $=$ $((1,m_1,w_2),(2,m_2,w_1),(2,w_1,m_2),(1,w_2,m_1))$. Then the set $T^\downarrow(\mathbf{x})$ is
     \begin{multline*} 
     \{((2,m_2,w_1),(2,w_1,m_2),(1,m_1,w_2),(1,w_2,m_1)),\\ 
     ((2,w_1,m_2),(2,m_2,w_1),(1,m_1,w_2),(1,w_2,m_1)),\\
     ((2,m_2,w_1),(2,w_1,m_2),(1,w_2,m_1),(1,m_1,w_2))\\
     ((2,w_1,m_2),(2,m_2,w_1),(1,w_2,m_1),(1,m_1,w_2))\}
     \end{multline*}
and the set $T^\downarrow_2(\mathbf{x})$ is $\{((2,m_2,w_1),(2,w_1,m_2)), ((2,w_1,m_2),(2,m_2,w_1))\}$.
     
The idea is to enumerate all vectors in $T_K^\downarrow  = \cup_{\mathbf{x}\in\mathcal{X}} T^\downarrow_K(\mathbf{x})$ without redundancy. The polynomiality of the approach follows from the fact that $|T_K^\downarrow| \leq (2n^2)^{K}$ because the number of distinct triples is upper bounded by $2n^2$. 
Note that we have:
$$
 \min_{\mathbf{x}\in\mathcal{X}} \GGI_{{\bm \lambda}}(\mathbf{d}(\mathbf{x})) =
  \min_{\mathbf{t}\in T_K^\downarrow} \GGI_{{\bm \lambda}}(\mathbf{t}) 
$$
because $\lambda_i=0$ for all $i>K$, where, by abuse of notation, we denote by $\GGI_{{\bm \lambda}}(\mathbf{t})$ the value of the GGI operator applied to the vector of dissatisfactions obtained from $\mathbf{t}$\footnote{Note that vector $\mathbf{t}\in T_K^\downarrow$ is incomplete as it only has $K$ components, but it is sufficient to apply the GGI operator because $\lambda_i=0$ for all $i>K$.}. Hence identifying an optimal GGI stable marriage will be performed by finding a vector $\mathbf{t}\in T_K^\downarrow$ minimizing the GGI operator and computing a corresponding stable marriage.

\begin{ex}
Coming back to the instance of Example \ref{exGILBERT}, assume that $K=2$ and that the disutility function is defined by $d(i)= i$. Then, our enumeration algorithm would produce the following set $\mathcal{T}^\downarrow_2$:

\begin{multline*} 
\{((10,w_{10},m_9),(7,w_{4},m_6)),((10,w_{10},m_9),(7,w_{5},m_7)),\\
((10,w_{10},m_9),(7,w_{6},m_5)),((10,w_{10},m_9),(7,w_{7},m_4)),\\
((10,w_{10},m_9),(5,m_{4},w_6)),((10,w_{10},m_9),(5,m_{5},w_7)),\\
((10,w_{10},m_9),(5,m_{6},w_5)),((10,w_{10},m_9),(5,m_{7},w_4)),\\
((9,w_{10},m_8),(5,m_{4},w_6)),((9,w_{10},m_8),(5,m_{5},w_7)),\\
((9,w_{10},m_8),(5,m_{6},w_5)),((9,w_{10},m_8),(5,m_{7},w_4)),\\
((9,w_{10},m_8),(5,m_{8},w_{10})),((9,w_{10},m_8),(5,m_{9},w_9)),\\
((9,w_{10},m_8),(5,m_{10},w_8))\}
\end{multline*}
For this instance, the optimal GGI value is therefore necessarily $9 \lambda_1 + 5 \lambda_2$. Note that, in most cases, the optimal GGI value depends on ${\bm \lambda}$ (it is not the case here because $(9,5)$ dominates componentwise all vectors of dissatisfactions obtained from $\mathcal{T}^\downarrow_2$).
\end{ex}

We now describe our enumeration algorithm. Algorithm \ref{main} builds set $T_K^\downarrow$ by induction using the following formula:
\begin{align*}
   T_0^\downarrow &= \{()\}\\
   T_k^\downarrow &= \{\mathbf{v} \circ t : \mathbf{v} \in T_{k-1}^\downarrow \text{ and } t \in T(\mathbf{v})\} \text{ where } T(\mathbf{v}) = \{ t_k : \mathbf{t} \in T_N^\downarrow \text{ s.t. } (t_1,\ldots, t_{k-1}) = \mathbf{v}\}
\end{align*}
The aim of Algorithm \ref{NextTriples} is to compute $T(\mathbf{v})$, i.e., the set of possible triples for the $k^{th}$ component of a vector in $T_k^\downarrow$ starting by the $(k-1)$-vector $\mathbf{v}$. The idea is to impose restrictions on the considered stable marriages so that the least satisfied agents as well as their matches correspond to the ones in $\mathbf{v}$. For this purpose, we impose mandatory rotations (set $\IN$) and forbidden rotations (set $\OUT$). 
Note that, each time a rotation is made mandatory (resp. forbidden), the set of its ancestors (resp. descendants), denoted by $\mathtt{Anc}(\rho)$ (resp. $\mathtt{Desc}(\rho)$), are also made mandatory (resp. forbidden) so that $\mathtt{IN}_{\mathbf{v}}$ (resp. $P\setminus \mathtt{OUT}_\mathbf{v}$) remains a closed set of rotations. 
For each triple $(d,a,a')$ belonging to $\mathbf{v}$, we ensure that agent $a$ is matched with agent $a'$ by making rotation $\rhoget(a,a')$ mandatory and $\rhobreak(a,a')$ forbidden (Lines 3--5). Additionally, to ensure that the $k$ least satisfied agents are indeed those involved in $\mathbf{v}$, we put a threshold on the dissatisfactions of the agents in $\mathcal{A}_{\overline{\mathbf{v}}} = \mathcal{M}\cup\mathcal{W}\setminus \{a:(d,a,a') \in \mathbf{v}\}$. Note that the set $\mathcal{A}_{\overline{\mathbf{v}}}$ is updated in Line 3. Let $d_{\min}(\mathbf{v})$ denote the dissatisfaction of the last triple in $\mathbf{v}$ (i.e., the lowest level of dissatisfaction in $\mathbf{v}$). The dissatisfactions of the agents in $\mathcal{A}_{\overline{\mathbf{v}}}$ should not be strictly greater than $d_{\min}(\mathbf{v})$. This condition is imposed by using again sets $\IN$ and $\OUT$. More precisely, given a rotation $\rho = (m_{i_0},w_{i_0}),\ldots,(m_{i_{r-1}},w_{i_{r-1}})$, we define $\dmaxw(\rho)=\max_{k=0,\ldots,r-1} d(w_{i_k},m_{i_k})$ the highest dissatisfaction of a woman involved in $\rho$ before $\rho$ is eliminated, and $\dmaxm(\rho)=\max_{k=0,\ldots,r-1} d(m_{i_k},w_{i_{k+1}})$ the highest dissatisfaction of a man involved in $\rho$ after $\rho$ is eliminated. To make sure that the agents in $\mathcal{A}_{\overline{\mathbf{v}}}$ have a dissatisfaction lower than or equal to $d_{\min}(\mathbf{v})$, we make mandatory (resp. forbidden) any rotation $\rho\in P\backslash \mathtt{OUT}_\mathbf{v}$ (resp. $P \backslash \mathtt{IN}_{\mathbf{v}}$) such that  $\dmaxw(\rho) > d_{\min}(\mathbf{v})$ (resp. $\dmaxm(\rho) > d_{\min}(\mathbf{v})$) (Lines 6--7, resp. Lines 8--9). The enumeration of the triples in $T(\mathbf{v})$ is performed by branching on the gender (man or woman) of the agent that will realize the $k^{th}$ highest dissatisfaction. We denote by $T_W(\mathbf{v})$ (resp. $T_M(\mathbf{v})$) the set of triples $(d,a,a') \in T(\mathbf{v})$ where $a \in \mathcal{W}$ (resp. $a \in \mathcal{M}$). We have of course $T_W(\mathbf{v}) \cup T_M(\mathbf{v}) = T(\mathbf{v})$. Algorithm \ref{NextWomen} enumerates the triples in $T_W(\mathbf{v})$ while Algorithm \ref{NextMen} enumerates the triples in $T_M(\mathbf{v})$ (Line 10 of Algorithm \ref{NextTriples}). The validity of the approach follows from the validity of Algorithms~\ref{NextWomen} and~\ref{NextMen}.

\paragraph{Validity of the approach.} The operations of Algorithms \ref{NextWomen} and \ref{NextMen} are similar. They proceed in the spirit of the algorithm proposed by \cite{gusfield1987three} for determining a minmax stable marriage. Let $\mathbf{x}_R$ denote the stable marriage corresponding to a set $R$ of rotations. Note that we have built sets $\mathtt{IN}_{\mathbf{v}}$ and $\mathtt{OUT}_\mathbf{v}$ such that if $R\cap\mathtt{IN}_{\mathbf{v}} = \mathtt{IN}_{\mathbf{v}}$ and $R\cap\mathtt{OUT}_\mathbf{v} = \emptyset$ then $\mathbf{v} \in T^\downarrow_{k-1}(\mathbf{x}_R)$. Furthermore, the special case $\mathbf{x}_{\mathtt{IN}_{\mathbf{v}}}$ (resp. $\mathbf{x}_{P\setminus \mathtt{OUT}_\mathbf{v}}$) is the stable marriage compatible with $\mathtt{IN}_{\mathbf{v}}$ and $\mathtt{OUT}_\mathbf{v}$ that satisfy most the men (resp. women) as it takes as few (resp. much) rotations as allowed by sets $\IN$ and $\OUT$. We only explain the operation of Algorithm \ref{NextWomen}, because the operation of Algorithm \ref{NextMen} is symmetric.

The aim of Algorithm \ref{NextWomen} is to enumerate all triples $(d,a,a')$ in $T_W(\mathbf{v})$.  Notably, we will enumerate these triples by nonincreasing values of $d$ by exploring carefully the set of stable marriages compatible with sets $\IN$ and $\OUT$. More precisely, at each iteration $i$ of the algorithm (loop {\tt while} in Line 5) we will consider a stable marriage $\mathbf{x}_i$ compatible with sets $\IN$ and $\OUT$ such that all women are always better off in $\mathbf{x}_{i}$ than in $\mathbf{x}_{i-1}$ for $i\neq 0$ (with at least one woman strictly better off). At each iteration, the new triples are found by looking at set $W_i$ that includes all women in $\mathcal{A}_{\overline{\mathbf{v}}}$ whose dissatisfaction can be ranked in $k^{th}$ position in $\mathbf{x}_{i}$, i.e., whose dissatisfaction is equal to $d^\downarrow_k(\mathbf{x}_{i})$ (Lines 3 and 13)\footnote{We recall that $d^\downarrow_k(\mathbf{x})$ denotes the $k^{th}$ component of vector $\mathbf{d}(\mathbf{x})$ when sorted by nonincreasing values.}.  

Obviously, for the women, the worst stable marriage compatible with $\IN$ and $\OUT$ is  $\mathbf{x}_{\IN}$. If no woman can be ranked in $k^{th}$ position w.r.t. stable marriage $\mathbf{x}_{\IN}$, then no woman can be ranked in $k^{th}$ position for \emph{any} stable marriage compatible with $\mathtt{IN}_{\mathbf{v}}$. Indeed, eliminating additional rotations would only increase the dissatisfactions of men and decrease the dissatisfactions of women. Otherwise the recurrence is initialized with $\mathbf{x}_0 = \mathbf{x}_{\IN}$ and stable marriage $\mathbf{x}_{i+1}$ is obtained from $\mathbf{x}_{i}$ by eliminating rotation $\rhobreak(m,w)$ (and all required ancestors) for all woman $w$ in $W_i$ so that their dissatisfactions are strictly decreased (Line 10). Loop {\tt while} stops if one of the following conditions occurs:
\begin{itemize}
\item if $W_i = \emptyset$, it means that only men can be ranked in $k^{th}$ position in $\mathbf{x}_i$; as eliminating rotations will only improve the situation of women and deteriorate the situation of men, we can safely conclude that all triples in $T_W(\mathbf{v})$ have been enumerated;
\item if at least  one rotation $\rhobreak(m,w)$ does not exist or is forbidden (i.e., $(m,w) \in \mathbf{x}_{P\backslash \OUT}$); indeed, in this case, we can conclude that it is not possible to find a triple in $T_W(\mathbf{v})$ with a dissatisfaction strictly less than the current value $d^\downarrow_k(\mathbf{x}_{i})$ (the boolean Flag is then set to True in Line 9).
\end{itemize}

\paragraph{Complexity analysis and proof of termination.} In Algorithm~\ref{NextWomen}, at every step $i$ of the {\tt while} loop, all agents in $W_i$ share the same dissatisfaction level $d^\downarrow_k(\mathbf{x}_{i})$. Furthermore, for all $i \neq 0$, we have that $d^\downarrow_k(\mathbf{x}_{i}) < d^\downarrow_k(\mathbf{x}_{i-1})$. As there are only $n$ dissatisfaction levels (corresponding to the $n$ possible ranks), the {\tt while} loop necessarily terminates in $O(n)$ iterations. The nested {\tt for} loop also terminates in $O(n)$ iterations because there can be at most $n$ women in $W_i$. All instructions inside the {\tt for} loop are in $O(1)$, except the instruction in Line 10 which is in $O(n^2)$ (the number of rotations is upper bounded by $n(n-1)/2$). Overall, Algorithm~\ref{NextWomen} is in $O(n^4)$. The analysis of Algorithm~\ref{NextMen} is similar. In Algorithm~\ref{NextTriples}, Lines 4 and 5 are in $O(n^2)$, hence the {\tt for} loop in Line 2 is in $O((k-1)n^2)$, therefore in $O(n^3)$ as $k \le 2n$. Lines 6--9 are in $O(n^4)$. Since we have shown that both calls in Line 10 are in $O(n^4)$, the overall complexity of Algorithm~\ref{NextTriples} is $O(n^4)$. Finally, the complexity of the three nested {\tt for} loops in Algorithm~\ref{main} is $O(\sum_{k=1}^K (2n^2)^{k-1} (n^4+2n^2))$ because:\\[1ex]
-- the cardinality of set $T_{(k-1)}^\downarrow$ in Line 4 is upper bounded by $(2n^2)^{k-1}$ (there are at most $2n^2$ triples, and $k-1$ components per vector of triples in $T_{(k-1)}^\downarrow$);\\
-- Line 5 is in $O(n^4)$; \\
-- Lines 6--7 are in $O(2n^2)$.\\[1ex]
Overall, the complexity of Algorithm~\ref{main} thus is $O(2^Kn^{2K+4})$.
 
\paragraph{Final remarks.} At the end of Algorithm \ref{main}, one obtains a set  
$T_K^\downarrow$ of vectors of triples. Within this set, one can choose a vector $\mathbf{v}^*$ which realizes:
$$
\min_{\mathbf{v}\in T_K^\downarrow} \GGI_{{\bm \lambda}}(\mathbf{v}) = \min_{\mathbf{x}\in\mathcal{X}} \GGI_{{\bm \lambda}}(\mathbf{d}(\mathbf{x})) .
$$
Given this vector $\mathbf{v}^*$, any stable marriage $\mathbf{x}^*$ such that $\mathbf{v}^* \in T_K^\downarrow(\mathbf{x^*})$ verifies
$$
\GGI_{{\bm \lambda}}(\mathbf{d}(\mathbf{x}^*))= \min_{\mathbf{x}\in\mathcal{X}} \GGI_{{\bm \lambda}}(\mathbf{d}(\mathbf{x})) .
$$

Given $\mathbf{v}^*$, it is easy to compute a stable marriage $\mathbf{x}^*$ such that $\mathbf{v}^* \in T_K^\downarrow(\mathbf{x}^*)$. In particular, $\mathbf{x}_{\mathtt{IN}_{\mathbf{v}^*}}$ (resp. $\mathbf{x}_{P \setminus \mathtt{OUT}_{\mathbf{v}^*}}$) is a best possible stable marriage for men (resp. women) where sets $\mathtt{IN}_{\mathbf{v}^*}$ and $\mathtt{OUT}_{\mathbf{v}^*}$ are generated in the same fashion as in Algorithm \ref{NextTriples} (Lines 1--9).  


\begin{algorithm}[]
\DontPrintSemicolon
\SetKwData{Left}{left}\SetKwData{This}{this}\SetKwData{Up}{up}
\SetKwFunction{Union}{Union}\SetKwFunction{FindCompress}{FindCompress}
\SetKwInOut{Input}{input}\SetKwInOut{Output}{output}
\Input{the GGISM instance and the value of $K$}
\Output{$T_K^\downarrow$}
 $T_0^\downarrow \leftarrow \{()\}$\\
 \For{$k = 1, \ldots, K$}{
 $T_k^\downarrow\leftarrow \emptyset$\\
      \For{$\mathbf{v} \in T_{k-1}^\downarrow$}{
      	$T\leftarrow \mathtt{NextTriples}(\mathbf{v},k)$\\
      	\For{$t \in T$}{
      		$T_k^\downarrow \leftarrow T_k^\downarrow \cup\{\mathbf{v}\circ t \}$
	}      
      }
  }
  \Return $T_K^\downarrow$
\caption{$\mathtt{Enumerate}$}
\label{main}
\end{algorithm}
\begin{algorithm}[]
\DontPrintSemicolon
\SetKwData{Left}{left}\SetKwData{This}{this}\SetKwData{Up}{up}
\SetKwFunction{Union}{Union}\SetKwFunction{FindCompress}{FindCompress}
\SetKwInOut{Input}{input}\SetKwInOut{Output}{output}
\Input{vector $\mathbf{v}$ of imposed triples, index $k$ of the next triple}
\Output{set $T$ of possible next triples}
 $\mathtt{IN}_{\mathbf{v}}\leftarrow \emptyset$; $\mathtt{OUT}_\mathbf{v}\leftarrow \emptyset$; $\mathcal{A}_{\overline{\mathbf{v}}}\leftarrow \mathcal{M} \cup \mathcal{W}$\\
 \For{$i = 1,\ldots,k-1$}{
 $(a,a',d) = v_i$, $\mathcal{A}_{\overline{\mathbf{v}}} \leftarrow \mathcal{A}_{\overline{\mathbf{v}}} \setminus \{ a \}$ \\
 $\mathtt{IN}_{\mathbf{v}} \leftarrow \mathtt{IN}_{\mathbf{v}}\cup \{ \rhoget(a,a')\}\cup \mathtt{Anc}(\rhoget(a,a')) $\\
 $\mathtt{OUT}_\mathbf{v} \leftarrow \mathtt{OUT}_\mathbf{v}\cup \{ \rhobreak(a,a')\}\cup \mathtt{Desc}(\rhobreak(a,a')) $}
      \For{$\rho \in P\backslash \mathtt{OUT}_\mathbf{v} \text{ s.t. } \dmaxw(\rho) > d_{\min}(v)$}{
      	$\mathtt{IN}_{\mathbf{v}} \leftarrow \mathtt{IN}_{\mathbf{v}}\cup \{\rho\}\cup \mathtt{Anc}(\rho)$ 	
	}
	\For{$\rho \in P \backslash \mathtt{IN}_{\mathbf{v}} \text{ s.t. } \dmaxm(\rho) > d_{\min}(v)$}{
      	$\mathtt{OUT}_\mathbf{v} \leftarrow \mathtt{OUT}_\mathbf{v}\cup \{\rho\}\cup \mathtt{Desc}(\rho)$ 	
	}      
\Return $\mathtt{NextWomen}(\mathtt{IN}_{\mathbf{v}},\mathtt{OUT}_\mathbf{v},k,\mathcal{A}_{\overline{\mathbf{v}}})$ $\cup$ $\mathtt{NextMen}(\IN,\mathtt{OUT}_\mathbf{v},k,\mathcal{A}_{\overline{\mathbf{v}}})$
\caption{$\mathtt{NextTriples}$}
\label{NextTriples}
\end{algorithm}
\begin{algorithm}[]
\DontPrintSemicolon
\SetKwData{Left}{left}\SetKwData{This}{this}\SetKwData{Up}{up}
\SetKwFunction{Union}{Union}\SetKwFunction{FindCompress}{FindCompress}
\SetKwInOut{Input}{input}\SetKwInOut{Output}{output}
\Input{set $\mathtt{IN}_{\mathbf{v}}$ and $\mathtt{OUT}_\mathbf{v}$ of mandatory and forbidden rotations, index $k$ of the next triple, set $\mathcal{A}_{\overline{\mathbf{v}}}$}
\Output{$T_W(\mathbf{v})$}
 Compute $\mathbf{x}_{\mathtt{IN}_{\mathbf{v}}}$ and $\mathbf{x}_{P\backslash \mathtt{OUT}_\mathbf{v}}$\\
$ T \leftarrow \emptyset$; $R \leftarrow \IN$; $i \leftarrow 0$; $\mathbf{x}_i \leftarrow \mathbf{x}_R$ \\ 
 $W_i \leftarrow \{w\in\mathcal{A}_{\overline{\mathbf{v}}}\cap \mathcal{W}: d(w,\mathbf{x}_{i}) = d^\downarrow_k(\mathbf{x}_{i})\}$\\
Flag $\leftarrow$ False\\ 
 \While{$W_i \neq \emptyset$}{
 \For{$w\in W_i$}{ 
 let $m$ be the match of $w$ in $\mathbf{x}_i$\\
 $T\leftarrow T\cup \{(d(w,m),w,m)\}$\\ 
 \lIf{$(m,w)\in\mathbf{x}_{P\backslash \mathtt{OUT}_\mathbf{v}}$}{Flag $\leftarrow$ True}
 \lElse{$R\leftarrow R \cup \rhobreak(m,w)\cup \mathtt{Anc}(\rhobreak(m,w))$}
 }
 \lIf{Flag}{\Return $T$}
 $i \leftarrow i+1$; $\mathbf{x}_i \leftarrow \mathbf{x}_R$\\
 $W_i \leftarrow \{w\in\mathcal{A}_{\overline{\mathbf{v}}}\cap \mathcal{W}: d(w,\mathbf{x}_{i}) = d^\downarrow_k(\mathbf{x}_{i})\}$
 }
\Return $T$
\caption{NextWomen}
\label{NextWomen}
\end{algorithm}
\begin{algorithm}
\DontPrintSemicolon
\SetKwData{Left}{left}\SetKwData{This}{this}\SetKwData{Up}{up}
\SetKwFunction{Union}{Union}\SetKwFunction{FindCompress}{FindCompress}
\SetKwInOut{Input}{input}\SetKwInOut{Output}{output}
\Input{set $\mathtt{IN}_{\mathbf{v}}$ and $\mathtt{OUT}_\mathbf{v}$ of mandatory and forbidden rotations, index $k$ of the next triple, set $\mathcal{A}_{\overline{\mathbf{v}}}$}
\Output{$T_M(\mathbf{v})$} 
 Compute $\mathbf{x}_{\mathtt{IN}_{\mathbf{v}}}$ and $\mathbf{x}_{P\backslash \mathtt{OUT}_\mathbf{v}}$ \\
 $ T \leftarrow \emptyset$; $R \leftarrow P \setminus \OUT$; $i \leftarrow 0$; $\mathbf{x}_i \leftarrow \mathbf{x}_R$ \\ 
 $M_i\leftarrow \{m\in\mathcal{A}_{\overline{\mathbf{v}}}\cap \mathcal{M}: d(m,\mathbf{x}_{i}) = d^\downarrow_k(\mathbf{x}_{i})\}$\\
Flag $\leftarrow$ False\\ 
 \While{$M_i\neq \emptyset$}{
 \For{$m\in M_i$}{ 
 let $w$ be the match of $m$ in $\mathbf{x}_i$\\
 $T \leftarrow T \cup \{(d(m,w),m,w)\}$\\ 
 \lIf{$(m,w) \in \mathbf{x}_{\mathtt{IN}_{\mathbf{v}}}$}{Flag $\leftarrow$ True}
 \lElse{$R\leftarrow R \setminus (\rhoget(m,w)\cup \mathtt{Des}(\rhoget(m,w)))$}
 }
 \lIf{Flag}{\Return $T$}
 $i \leftarrow i+1$; $\mathbf{x}_i \leftarrow \mathbf{x}_R$\\
  $M_i\leftarrow \{m\in\mathcal{A}_{\overline{\mathbf{v}}}\cap \mathcal{M}: d(m,\mathbf{x}_{i}) = d^\downarrow_k(\mathbf{x}_{i})\}$
 }
\Return $T$
\caption{NextMen}
\label{NextMen}
\end{algorithm}

\section{Conclusion}

In this paper, we have shown that the minimization of a Generalized Gini Index (GGI) of the dissatisfactions of men and women in a stable marriage problem is an NP-hard problem. Then, we have proposed a polynomial time 2-approximation algorithm for the problem, based on a rounding of the optimal solution to the linear programming relaxation of the problem. Lastly, we have shown that minimizing a GGI of the dissatisfactions of men and women in a stable marriage is in the class XP with respect to the number of strictly positive weights in the GGI operator.

For future works, following \cite{aziz2017random}, it could be worth investigating the randomized version of the GGI stable marriage problem. By \emph{randomized}, we mean that we consider mixed stable marriages, and not only deterministic stable marriages. A mixed stable marriage is a probability distribution over stable marriages. This enlargement of the set of feasible solutions could make it possible to enhance the optimal GGI value (where the GGI operator is applied to the vector of expected dissatisfactions of the agents). Note that the relaxed solution  we compute in the first step of the 2-approximation algorithm proposed in Section~\ref{sec:approximation} can be converted into a mixed stable marriage by using a trick proposed by \cite{teo1998geometry}. It turns out that the obtained approach returns an optimal marriage for the randomized variant of the GGI stable marriage problem. A more thorough investigation of the randomized GGI stable marriage problem is underway.

\bibliographystyle{named}
 \bibliography{algorithmica2018}

\end{document}